\documentclass{article}

\usepackage{microtype}
\usepackage{graphicx}
\usepackage{subcaption}
\usepackage{booktabs} % for professional tables
\usepackage{siunitx}
\usepackage{dsfont,mathrsfs}
\usepackage{enumitem}
\usepackage{hyperref}

\usepackage{multirow}

\usepackage[preprint]{icml2026}

\usepackage{amsmath}
\usepackage{amssymb}
\usepackage{mathtools}
\usepackage{amsthm}

\usepackage[capitalize,noabbrev]{cleveref}

\theoremstyle{plain}
\newtheorem{theorem}{Theorem}[section]
\newtheorem{proposition}[theorem]{Proposition}
\newtheorem{lemma}[theorem]{Lemma}
\newtheorem{corollary}[theorem]{Corollary}
\theoremstyle{definition}
\newtheorem{definition}[theorem]{Definition}

\theoremstyle{remark}
\newtheorem{remark}[theorem]{Remark}

\newcommand{\smallvec}[2]{%
  \left[ \begin{smallmatrix} 
    #1\vphantom{\pi^1(a^1)} \\ #2\vphantom{\pi^1(a^1)} 
  \end{smallmatrix} \right]%
}

\makeatletter
\DeclareRobustCommand{\cev}[1]{%
  {\mathpalette\do@cev{#1}}%
}
\newcommand{\do@cev}[2]{%
  \vbox{\offinterlineskip
    \sbox\z@{$\m@th#1 x$}%
    \ialign{##\cr
      \hidewidth\reflectbox{$\m@th#1\vec{}\mkern4mu$}\hidewidth\cr
      \noalign{\kern-\ht\z@}
      $\m@th#1#2$\cr
    }%
  }%
}
\makeatother
\usepackage{xcolor}
\newcommand*\blue{\textcolor{blue}}

\newcommand\ALPHABET{\mathsf}
\newcommand\reals{\mathds{R}}

\newcommand\naturalnumbers{\mathds{N}}

\newcommand\PR{\mathds{P}}
\newcommand\EXP{\mathds{E}}
\newcommand\REXP{\mathcal{R}}

\newcommand*\GREEDY{\mathcal{G}}

\newcommand*\risk{\mathcal{R}}

\DeclareMathOperator*{\argmax}{arg\,max}

\usepackage[textsize=tiny]{todonotes}

\icmltitlerunning{Multi-agent RS-CPI}
\newcommand\important[1]{\textbf{\textit{#1}}}

\begin{document}

\twocolumn[
  \icmltitle{Risk-seeking conservative policy iteration with agent-state based policies for Dec-POMDPs with guaranteed convergence}



  \icmlsetsymbol{equal}{*}

  \begin{icmlauthorlist}
    \icmlauthor{Amit Sinha}{mcgill}
    \icmlauthor{Matthieu Geist}{esp}
    \icmlauthor{Aditya Mahajan}{mcgill}
  \end{icmlauthorlist}

  \icmlaffiliation{mcgill}{Department of Electrical and Computer Engineering, McGill University, Montreal, Quebec}
  \icmlaffiliation{esp}{Earth Species Project}

  \icmlcorrespondingauthor{Amit Sinha}{amit.sinha@mail.mcgill.ca}

  \icmlkeywords{Multi-agent teams, risk-awareness, Dec-POMDPs, Finite-state controllers, partial observability}

  \vskip 0.3in
]



\printAffiliationsAndNotice{}  % no special notice (required even if empty)

\begin{abstract}
Optimally solving decentralized decision-making problems modeled as Dec-POMDPs is known to be NEXP-complete. These optimal solutions are policies based on the entire history of observations and actions of an agent. However, some applications may require more compact policies because of limited compute capabilities, which can be modeled by considering a limited number of memory states (or agent states). While such an agent-state based policy class may not contain the optimal solution, it is still of practical interest to find the best agent-state policy within the class. We focus on an iterated best response style algorithm which guarantees monotonic improvements and convergence to a local optimum in polynomial runtime in the Dec-POMDP model size. In order to obtain a better local optimum, we use a modified objective which incentivizes risk-seeking alongside a conservative policy iteration update. Our empirical results show that our approach performs as well as state-of-the-art approaches on several benchmark Dec-POMDPs, achieving near-optimal performance while having polynomial runtime despite the limited memory. We also show that using more agent states (a larger memory) leads to greater performance. Our approach provides a novel way of incorporating memory constraints on the agents in the Dec-POMDP problem.
\end{abstract}

\section{Introduction}

Decentralized Partially Observable Markov Decision Processes (Dec-POMDP) model multi-agent sequential decision problems where agents with asymmetric partial observation must coordinate their actions to optimize a shared cumulative reward. Such models arise in various applications including autonomous drone fleets \cite{floriano2019planning}, network load balancing \cite{yao2022multi}, and multi-robot co-ordination \cite{omidshafiei2017decentralized}. Consequently, Dec-POMDPs have been studied extensively across multiple research disciplines including Economics and Organizational Behavior as team theory \cite{Marschak:1955,Radner:1962,MarschakRadner:1972} and Systems and Control as decentralized stochastic control \cite{witsenhausen1971separation,witsenhausen1973standard,Sandell:1978,nayyar2013decentralized} Optimally solving Dec-POMDPs presents three fundamental challenges:

\important{(i)~Information decentralization:}
Let $I^i_t$ denote the information available to agent $i$ makes decisions at time $t$. Typically, in Dec-POMDPs agents lack perfect recall ($I^i_t \not\subset I^i_{t+1}$) and possess asymmetric information ($I^i_t \neq I^j_t$ for $i \neq j$). Thus Dec-POMDPs have \emph{non-classical information structure}~\cite{witsenhausen1971separation}, which precludes the use of standard dynamic programming (DP) due to the absence of a single decision-maker with perfect recall. 

\important{(ii)~Partial observations with agent states:}
Agents receive only partial observations of the environment state and maintain a compressed summary of past observations and actions as a local internal state (called agent state). Due to fixed memory of the agent state, agents must perform lossy state updates, selectively discarding non-salient details to prioritize information critical for decision-making.

\important{(iii)~Coordination among agents.}
Agents must coordinate their policies not only to maximize immediate rewards but also to strategically signal information, thereby facilitating the optimization of future joint actions. Because an update in one agent's policy alters the best response of others, the learning task becomes a joint optimization problem. Consequently, agents must coordinate their learning processes within the joint policy space rather than optimizing each policy independently.

% The updates of one agent's policy affect the other agent's best course of action. Thus, this is a joint optimization problem where the policies of all the agents should involve some co-ordination in the optimization process rather than optimizing for each policy independently.

% They key challenges of Dec-POMDPs are $(i)$ partial observability: agents have to consider the entire local history of their individual actions and observations, which is an information set that increases exponentially with time, $(ii)$ co-ordination: agents need to choose their policies based on the policies of other agents, which makes it a joint optimization problem, and $(iii)$ decentralization of information: different agents have different pieces of information relevant to solving the task at hand.

The complexity of these challenges is reflected in the theoretical hardness of the problem: solving a Dec-POMDP optimally is NEXP-complete \cite{bernstein2002complexity}, and this hardness extends even to finding approximate solutions \cite{rabinovich2003complexity}. Various approaches have been taken in the literature to address these challenges.

% \red{[This needs to be polished...but I feel that we need to make this argument.]}

\paragraph{Related Work:\footnote{A more detailed literature review is given in Appendix~\ref{app:related-work}.}} On the theoretical side, different algorithms for computing optimal solutions have been proposed. The designer's approach \cite{MahajanPhD, witsenhausen1973standard} can be used to used to obtain a dynamic program to obtain the optimal solution for a general finite-state controller. However, the size of such a dynamic program is exponential in the size of the Dec-POMDP spaces and so, it does not scale well to larger settings. In some Dec-POMDPs, there may be some structure in the form of common information between agents that can be exploited by ideas like the common information approach \cite{nayyar2013decentralized, tang2023novel}, but structural assumption of common information does not always hold. As a result these ideas only work for small scale environments.

% Exhaustively enumerating all possible solutions was done using Multi-agent A$^{\ast}$ (MAA$^{\ast}$) \cite{szer2012maa}, which was then generalized by Generalized Multi-agent A$^{\ast}$ (GMAA$^{\ast}$) \cite{oliehoek2008optimal}. In order to make the exponential search more tractable, GMAA$^{\ast}$-ICE \cite{oliehoek2013incremental} proposes lossless incremental clustering to group histories together that satisfy the notion of probabilistic equivalence, and incremental expansion which selectively expands the tree one node at a time.

Feature-based heuristic search value iteration (FB-HSVI) \cite{dibangoye2016optimally} involve the construction of a dynamic program similar to the designer's approach (referred to as the ``occupancy MDP'') and the exponential complexity can be dealt with by clustering and a sliding window of the history. \citet{koops2024approximate} also uses follows a similar outline in their Policy-finding multi-agent A$^\ast$ (PF-MAA$^{\ast}$) algorithm. These approaches involve finding near-optimal history-based solutions. One important limitation of these works is that they are only able to use a small window size, because otherwise the exponential growth of trajectories becomes intractable to deal with.

In certain applications, the agents have a limited number of memory states, which is commonly referred to as an ``agent state'' \cite{dong2022simple}. It then becomes important to find the best possible solution with the class of \emph{agent-state based policies} rather than the class of history-based policies. Ideas like the designer's approach only work well for small problems, while the current history-based approaches (FB-HSVI and PF-MAA$^\ast$) do not currently provide agent-state based policies. \important{Our main contribution is to propose an algorithm which scales well and has theoretical convergence guarantees for finding good agent-state based policies.}

\paragraph{Main Idea:} We are interested in obtaining a locally optimal solution in polynomial runtime. Given a cost function $J$, a locally optimal joint policy of $N$ agents $\pi^\star = (\pi^{\star 1}, \pi^{\star 2}, \cdots, \pi^{\star N})$ is defined as:
\begin{align} \label{eq:basic-PBPO}
    \pi^{\star i} \in \argmax_{\pi^{\star i} \in \Pi^i} J(\pi^{\star}), \quad \forall i.
\end{align}
We consider an approach that uses iterated best response to obtain a local optimum. Given a starting joint policy of $N$ agents $ \pi = ( \pi^1, \pi^2, \cdots, \pi^N)$, this involves repeatedly applying the following update sequentially for each agent $i \in \{1, \cdots, N\}$:
\begin{align} \label{eq:iterated-BR}
    \pi^i \in \argmax_{\pi^i \in \Pi^i} J(\pi^i, \pi^{-i}),
\end{align}
where $\pi^{-i}$ denotes the policy of all agents except $i$. Repeatedly applying such an update (with ties broken with a fixed rule) provably leads to monotonic improvements and convergence to a local optimum. This fundamental idea is similar to the line of work on Joint Equilibrium-based Search for Policies (JESP) \cite{nair2003taming}.

In particular, we combine the following ideas that empirically allow convergence to better local optima (higher $J$):
\begin{enumerate}
    \item \textbf{Sequential Equilibrium}: Keeping the policies of all agents fixed except the one we are optimizing for allows us to write a novel dynamic program in terms of the agent state which enables guaranteed monotonic improvements and convergence.
    \item \textbf{Risk-seeking objective} \cite{shen2013risk}: Using such an objective allows states with high rewards but low probabilities to be perceived with a larger weight than those with lower rewards. This pushes agents to select actions to maximize these marginal probabilities leading to better local optima. We provide a concrete example illustrating this in the next section.
    \item \textbf{Conservative-policy iteration (CPI)} \cite{kakade2002approximately}: Without conservatism the policies generated using iterated best response are deterministic after the first iteration. Consequently, some parts of the state space have zero or negligible marginal distribution, preventing further improvement regardless of how actions are chosen in those states.
\end{enumerate}
We refer to our approach as Risk-Seeking Conservative Policy Iteration (RS-CPI).

Our main contributions are summarized as follows:
\begin{enumerate}[noitemsep, topsep=0pt]
    \item We provide a scalable algorithm RS-CPI that combines sequential equilibria for agent-state based policies with risk-seeking and conservatism while retaining convergence and monotonic improvement properties. \looseness=-1
    \item We demonstrate comparable performance of RS-CPI with the state-of-the-art algorithms from the literature which are (approximately) optimal on several standard Dec-POMDPs.
    \item We show that the performance of RS-CPI improves as the allowed agent-state size is increased.
    \item We also provide ablations that show the adding risk-seeking and conservatism is critical to performance.
\end{enumerate}

\subsection{Motivating risk-seeking incentives}
Standard Markov decision theory typically focuses on maximizing the expected cumulative return, representing a \emph{risk-neutral} objective. In contrast, \emph{risk-aware} formulations evaluate the expected utility of a stochastic return $X$, denoted by $\EXP[U(X)]$. When the utility function $U$ is concave, Jensen's inequality implies $\EXP[U(X)] \le U(\EXP[X])$; in this case, the decision-maker prefers a deterministic return over a random return with the same expectation, a behavior termed \emph{risk-averse}. Conversely, if $U$ is convex, the inequality is reversed and the agent is \emph{risk-seeking}. The special case where $U$ is linear corresponds to the standard risk-neutral paradigm. 

A commonly used risk-aware measure is the \emph{entropic risk} utility, defined as $U(x) = e^{\lambda x}$, where $\lambda \in \mathbb{R}$ parameterizes the degree of risk sensitivity. While the exponential function is always convex, the agent's objective depends on the sign of $\lambda$: a risk-seeking agent ($\lambda > 0$) maximizes $\mathbb{E}[e^{\lambda X}]$, whereas a risk-averse agent ($\lambda < 0$) minimizes $\mathbb{E}[e^{\lambda X}]$. To ensure the risk-aware metric remains on the same scale as the original rewards, we apply the inverse transformation $U^{-1}(\cdot)$ after taking the expectation, yielding the \emph{certainty equivalent} return:
\begin{equation*}
    \REXP(X) = \tfrac{1}{\lambda} \log \EXP[ e^{\lambda X}],
\end{equation*}
where $\REXP(X)$ represents the deterministic return that provides the same utility as the stochastic return $X$.

Risk awareness changes the behavior of the agent.  Consider an agent choosing among $m$ mutually exclusive outcomes with with different success probabilities. A risk-seeking agent (convex $U$) will typically bet on a single outcome that offers the highest potential utility, effectively gambling on the best-case scenario. Conversely, a risk-averse agent (concave $U$) will spread their bet across multiple outcomes, sacrificing maximum potential returns to ensure a more consistent and robust baseline performance. For this reason, there is a rich literature on risk \emph{sensitive} Markov decision processes \cite{shen2013risk}.

Our core insight is that risk-\emph{seeking} behavior naturally encodes a form of \emph{optimism} in the face of uncertainty. By prioritizing high-potential outcomes, this optimistic bias enables iterated best-response methods to escape poor local optima and converge toward superior local solutions.

We consider the iterated best response update \eqref{eq:iterated-BR}, similar to JESP~\cite{nair2003taming}. Iterated best response is guaranteed to converge to a local optimum \eqref{eq:basic-PBPO}. In order to improve the local optimum (higher $J$), we use a risk-seeking modified objective. We provide the following example to illustrate how a risk-seeking objective can converge to a better local optimum.

\textbf{Example}: Consider a single stage two-player game with the same reward function (mentioned in table below). Agent $1$ has actions: $\{ a^1, b^1 \}$ and Agent $2$ has actions $\{ a^2, b^2 \}$.
\[
\begin{array}{cc|cc}
 & & \multicolumn{2}{c}{\text{Agent $2$}} \\
 & & a^2 & b^2 \\
\hline
\text{Agent $1$} & a^1 & 2 & -10 \\
 & b^1 & -10 & 6
\end{array}
\]
Suppose we denote the policy of Agent $1$ by $\pi^1(a^1)$ and that for Agent $2$ as $\pi^2(a^2)$. Clearly the global optimum (best local optimum) is at $\smallvec{\pi^1(a^1)}{\pi^2(a^2)} = \smallvec{0}{0}$ with a reward of $+6$. There is also an inferior local optimum at $\smallvec{\pi^1(a^1)}{\pi^2(a^2)} = \smallvec{1}{1}$ with a reward of $+2$.

A JESP style approach typically involves randomly initializing the policy for all agents and then applying iterative best response for each agent one by one while keeping the other agents' policies fixed at the most recent policy update. For example, if we start with the initialization $\smallvec{\pi^1(a^1)}{\pi^2(a^2)} = \smallvec{0.9}{0.9}$, then we fix Agent $2$'s policy and update Agent $1$'s policy and then vice-versa. This process is repeated until convergence, which follows the trajectory:
\begin{align*}
    \smallvec{0.9}{0.9} \to \smallvec{1}{0.9} \to \smallvec{1}{1} \to \smallvec{1}{1} \to \cdots,
\end{align*}
which gives us the inferior local optimum. Similarly, if we start with $\smallvec{\pi^1(a^1)}{\pi^2(a^2)} = \smallvec{0.1}{0.1}$:
\begin{align*}
    \smallvec{0.1}{0.1} \to \smallvec{0}{0.9} \to \smallvec{0}{0} \to \smallvec{0}{0} \to \cdots,
\end{align*}
it gives us the global optimum. This shows that the initialization determines the equilibrium we converge to when using iterated best response. In the first case, the best action combination has a low probability of $0.01$ which is the reason for convergence to the local optimum.

We can remedy this situation by considering the application of risk-seeking rewards because it assigns a higher weight to higher rewards (optimistic), whereas the usual expected reward is risk-neutral. For example, an entropic risk map uses the logsumexp function with a positive temperature $\lambda$. If we then replace the expectation with this logsumexp function in the iterated response for the first case, we get:
\begin{align*}
    \smallvec{0.9}{0.9} \to \smallvec{0}{0.9} \to \smallvec{0}{0} \to \smallvec{0}{0} \to \cdots,
\end{align*}
and for the second case:
\begin{align*}
    \smallvec{0.1}{0.1} \to \smallvec{0}{0.1} \to \smallvec{0}{0} \to \smallvec{0}{0} \to \cdots,
\end{align*}
which both converge to the global optimum! Thus, it is possible to use risk-seeking rewards to reach better local optima.

\textbf{Notation:} Consider a discrete time system that operates from step $1$ to step $T$. We use $\cev x_t$ to denote all quantities up to (and excluding) time $t$, i.e., $x_{1:t-1}$.  We use $\vec x_t$ to denote all quantities after (and excluding) time $t$, i.e., $x_{t+1:T}$. We use capital letters $X$ to denote random variables and small letters $x$ to denote their realizations. We use sans-serif font $\ALPHABET{X}$ to denote the corresponding spaces. We use subscripts for the time index $t$ and superscripts for the agent index $i$ as $x^i_t$. We use the notation $x_t$ without any superscript to denote the quantity for all agents, i.e., $x_t \coloneqq (x^1, \cdots, x^N)$, where $N$ is the number of agents. We also use bold $\boldsymbol{x^i}$ to denote the quantity over time for agent $i$, i.e., $\boldsymbol{x^i} \coloneqq (x^i_1, \cdots, x^i_T)$. Similarly, we also use $\boldsymbol{x} \coloneqq (x_1, \cdots, x_T)$.

\section{Background}

\subsection{Dec-POMDPs with agent-state based policies} \label{sec:background}

Decentralized partially observable Markov decision processes (Dec-POMDP) model 
sequential decision-making problems involving multiple cooperative agents acting 
under partial observability and asymmetric information.

\textbf{Dec-POMDP Model.} A finite-horizon Dec-POMDP is defined by the tuple
\(
\mathcal{D} = \langle \ALPHABET{N}, \ALPHABET{S}, \ALPHABET{A}, \ALPHABET{Y}, P, r, \zeta_1, T \rangle,
\)
where:

\begin{itemize}[noitemsep, topsep=0pt,leftmargin=1em]
    \item $\ALPHABET{N} \coloneqq \{1,\cdots,N\}$ denotes the set of agents. %, with $N = |\ALPHABET{N}|$.
    \item $\ALPHABET{S}$ is a finite set of environment states. 
    \item $\ALPHABET{A} = \ALPHABET{A}^1 \times \cdots \times \ALPHABET{A}^N$ is the joint action space, 
    where $\ALPHABET{A}^i$ denotes the action set of agent $i \in \ALPHABET{N}$. 
    \item $\ALPHABET{Y} = \ALPHABET{Y}^1 \times \cdots \times \ALPHABET{Y}^N$ is the joint observation space,
    where $\ALPHABET{Y}^i$ denotes the observation set of agent $i \in \ALPHABET{N}$.
    \item $P$ is the joint transition–observation distribution, i.e.,
    \[
    \PR(s_{t+1}, y_{t+1} \mid s_{1:t}, y_{1:t}, a_{1:t}) = P(s_{t+1}, y_{t+1} \mid s_t, a_t)
    %P(s', y' \mid s, a) = \mathbb{P}(s_{t+1}, y_{t+1} = s', y' \mid s_t, a_t = s, a),
    \]
    where $s_t \in \ALPHABET{S}$ denotes the environment state, $y^i_t$ and $a^i_t$ denote the observation and action of agent~$i$, and $y_t = (y^1_t, \dots, y^n_t)$ and $a_t = (a^1_t, \dots, a^n_t)$ denotes the observation and action of all agents, all at time~$t$
    \item $r : \ALPHABET{S} \times \ALPHABET{A} \rightarrow \reals$ is the shared per-step reward function and $r(s_t,a_t)$ is the reward received at time~$t$.
    \item $\zeta_1 \in \Delta(\ALPHABET{S} \times \ALPHABET{Y})$ is joint distribution of the initial state and observations.
    \item $T \in \naturalnumbers$ is the finite decision horizon.
\end{itemize}

% At each time step $t \in \{1, \cdots, T\}$, the system is in state $s_t \in \ALPHABET{S}$.
% Each agent $i \in \ALPHABET{N}$ selects an action $a^i_{t} \in \ALPHABET{A}^i$
% based on its local action and observation history. The agents then receive a common reward
% $r(s_t, a_t)$, where $a_t = (a^1_{t}, \cdots, a^N_{t})$.
% The environment transitions to a new state $s_{t+1}$ and each agent receives an observation
% $y^i_{t+1} \in \ALPHABET{Y}^i$ according to $P$. 

% Each agent $i \in \ALPHABET{N}$ uses a policy $\phi^i_t$ to select the action $a^i_t$ based on its local action and observation history $(a^i_{1:t-1}, y^{i}_{1:t})$, i.e., $\phi^i_t \colon {\ALPHABET{A}^i}^{t-1} \times {\ALPHABET{Y}^{i}}^{t} \to \Delta(\ALPHABET{A}^i)$. We refer to $\boldsymbol{\phi}^i = (\phi^i_1, \cdots, \phi^i_T)$ as the policy of agent~$i$ and $\boldsymbol{\phi} = (\boldsymbol{\phi}^1, \cdots, \boldsymbol{\phi}^N)$ as the joint policy of all agents.

% The performance of a policy $\boldsymbol{\phi}$ is given by the expected total cumulative reward over horizon~$T$:
% \begin{align*}
%     J(\boldsymbol{\phi}) \coloneqq \EXP^{\boldsymbol{\phi}} \Biggl[ \sum_{t=1}^{T} r(S_t, A_t) \Biggr].
% \end{align*}

% The best PBPO joint policy is given by:
% \begin{align*}
%     \boldsymbol{\phi^{\star}} \coloneqq \sup_{\boldsymbol{\phi} \in \boldsymbol{\Phi}} J(\boldsymbol{\phi}).
% \end{align*}

%\subsection{Using agent-state based policies}

\textbf{Agent-state based policies.} We assume that instead of using their entire history of observations and actions, agents act based on an internal state which is commonly referred to as \emph{agent state} in the literature~\cite{dong2022simple}. 

%An alternative to using large history-based policies is to use much more compact memory limited finite-state controllers where the finite-states are commonly referred to as \emph{agent states} in the literature \cite{dong2022simple}. This could potentially result in the loss of optimality, however, it is of practical importance to work with finite memory constraints.

% If the time horizons are not small, then the possible policy space which uses local histories as inputs becomes prohibitively large. This also combines with the possible policies of other agents and exacerbates the problem.Instead one can also consider the use of finite-state controllers where the finite-states are also referred to as agent states and are capable of representing the local history of the agent in a much more compact way, since several different histories can often be mapped to the same agent state. This could potentially result in the loss of optimality, however, it is still of interest to obtain the best joint policy in the class of agent states. Due to the structure in several Dec-POMDPs, it is often possible to reach very close to the optimal performance with a small number of agent states.

Let $\ALPHABET{Z} = \ALPHABET{Z}^1 \times \cdots \times \ALPHABET{Z}^N$ be the space of agent states of all agents, where $\ALPHABET{Z}^i$ is the agent-state space of agent $i \in \ALPHABET{N}$. An \textbf{agent-state based policy} for agent $i$ is $\boldsymbol{\pi}^i = (\pi^i_1, \dots, \pi^i_T)$, where $\pi^i_t \colon \ALPHABET{Y}^i \times \ALPHABET{Z}^i \to \Delta(\ALPHABET{A}^i \times \ALPHABET{Z}^i)$. When following such a policy, at each time each agent samples its action and next agent state as a function of its current observation and agent state, i.e., $(a^i_t, z^i_t) \sim \pi^i_t(y^i_t, z^i_{t-1})$. We denote the joint policy of all agents as $\boldsymbol{\pi} = (\boldsymbol{\pi}^1, \cdots, \boldsymbol{\pi}^N)$ or more compactly as $\boldsymbol{\pi} = (\boldsymbol{\pi}^i, \boldsymbol{\pi}^{-i})$, when emphasizing the perspective of agent~$i$. 

Given a joint policy $\boldsymbol{\pi} \in \boldsymbol{\Pi}$, the system evolves as follows. Each agent~$i$ initializes its agent state $z^i_0$ independently at random according to some arbitrary distribution $\phi^i$. The initial state and observation $(s_1, y_1) \sim \zeta_1$. Then, for each time, agent~$i$ chooses $(a^i_t, z^i_t) \sim \pi^i_t(y^i_t, z^i_{t-1})$. The team of all agents receives the reward $r(s_t, a_t)$ and the environment state and agent's observation evolve as $(s_{t+1}, y_{t+1}) \sim P(\cdot | s_t, a_t)$.

We denote the set of all agent-state based policies for agent~$i$ by $\boldsymbol{\Pi}^i$, with $\boldsymbol{\Pi} = \prod_{i \in \ALPHABET{N}} \boldsymbol{\Pi}^i$ and $\boldsymbol{\Pi}^{-i} = \prod_{j \neq i} \boldsymbol{\Pi}^j$. The performance of a joint policy $\boldsymbol{\pi} \in \boldsymbol{\Pi}$ is given by the expected total cumulative reward over horizon~$T$:
\[
    J(\boldsymbol{\pi}) = \EXP^{\boldsymbol{\pi}} \biggl[ \sum_{t=1}^{T} r(S_t, A_t) \biggr].
\]

There are two notions of optimality in Dec-POMDPs. A joint policy $\boldsymbol{\pi}^\star \in \boldsymbol{\Pi}$ is called \textbf{globally optimal} if
\[
    J(\boldsymbol{\pi}^\star) \ge J(\boldsymbol{\pi}), \quad \forall \boldsymbol{\pi} \in \boldsymbol{\Pi}.
\]
A joint policy $\boldsymbol{\pi}^\star \in \boldsymbol{\Pi}$ is called \textbf{person-by-person optimal (PBPO)} if 
\[
    J(\boldsymbol{\pi}^{i,\star}, \boldsymbol{\pi}^{-i,\star}) \ge J(\boldsymbol{\pi}^i, \boldsymbol{\pi}^{-i,\star}), \quad \forall \boldsymbol{\pi}^i \in \boldsymbol{\Pi}^i, \forall i \in \ALPHABET{N}.
\]
Thus, PBPO corresponds to a Nash equilibrium of Dec-POMDP where all agents share a common utility function.

\subsection{Risk-aware MDPs with entropic risk utility}

In this section, we review the basic theory of risk-aware MDPs with entropic risk utility~\cite{howard1972risk}. Consider an MDP
$\mathcal{M} = \langle \ALPHABET{S}, \ALPHABET{A}, P, r, \zeta_1, T \rangle$, where the the components are similar to those defined for Dec-POMDPs. To distinguish with the notation for Dec-POMDPs, we use $\boldsymbol{\psi} = (\psi_1, \dots, \psi_T)$ to denote the policy, where $\psi_t \colon \ALPHABET{S} \to \Delta(\ALPHABET{A})$. The performance of a policy $\boldsymbol{\psi}$ from an initial time $t$ to final time $T$ is defined as
\begin{align*}
    J^{\risk}_{1:T}(\boldsymbol\psi) &= 
    \REXP\biggl[ \sum_{t=1}^T r(S_t, A_t) \biggr] \\
    &= 
    \tfrac{1}{\lambda} \log \EXP^{\psi} \biggl[  e^{\lambda \sum_{\tau=t}^T r(S_\tau, A_\tau) } \biggr].
\end{align*}
% We want to find a policy $\boldsymbol{\psi}$ that maximizes $J_T(\psi)$, i.e.,
% \begin{align*}
%     \psi^{\star} = \argmax_{\psi} J_T(\psi).
% \end{align*}

% \subsection{Dynamic Programming Decomposition}

% Consider the two steps of:
% \begin{enumerate}
%     \item Policy Evaluation
%     \item Policy Optimization
% \end{enumerate}

For a given policy, the risk-aware expected performance can be computed using dynamic programming.
Define $V_{T+1}^{\boldsymbol{\psi}}(s) \coloneqq 0$, for any state $s$. Then for any $t \in \{ T, \cdots, 1\}$ recursively define:
\begin{align*}
    Q^{\boldsymbol{\psi}}_t(s,a) &= \frac{1}{\lambda} \log \biggl( \sum_{s_{+}} P(s_{+} | s,a) 
    e^{\lambda [ r(s,a) + V^{\boldsymbol{\psi}}_{t+1} (s_{+})] } \biggr), \\ 
    V^{\boldsymbol{\psi}}_t(s) &= \frac{1}{\lambda} \log \biggl( \sum_{a} \psi_t(a \mid s) e^{\lambda Q^{\boldsymbol{\psi}}_t(s,a)}\biggr).
\end{align*}
Then,
\[
  J^{\risk}_{t:T}(\boldsymbol{\psi}) = \frac{1}{\lambda} \log \biggl( \sum_{s,a} \PR^{\cev \psi_t}(S_t = s) e^{\lambda V^{\boldsymbol{\psi}}_t(s)}\biggr).
\]

% \begin{align*}
%     V^{\boldsymbol{\psi}}_t(s) &= \frac{1}{\lambda} \log \left( \sum_{s_{+}, a} P(s_{+} \mid s,a) \psi_t(a \mid s) \exp \left(\lambda [ r(s,a) + V^{\boldsymbol{\psi}}_{t+1} (s_{+}) ] \right)  \right)
% \end{align*}

The optimal value functions can be computed as follows. Define $V^{\star}_{T+1}(s) \coloneqq 0$, for any state $s$. Then for any $t \in \{ T, \cdots, 1\}$ recursively define:
\begin{align*}
    Q^{\star}_t(s,a) &= \frac{1}{\lambda} \log \biggl( \sum_{s_{+}} P(s_{+} | s,a) 
    e^{\lambda [ r(s,a) + V^{\star}_{t+1} (s_{+})] } \biggr), \\ 
    V^{\star}_t(s) &= \max_{\psi_t} \frac{1}{\lambda} \log \biggl( \sum_{a} \psi_t(a \mid s) e^{\lambda Q^{\star}_t(s,a)}\biggr),
    \\
    &= \max_{a} Q^{\star}_t(s,a),
\end{align*}
where the last equality follows from the monotonicity of $e^{\lambda x}$. Thus, the optimal policy $\boldsymbol \psi^{\star} = (\psi^{\star}_1, \dots, \psi^{\star}_T)$ is deterministic and satisfies $\psi^\star_t(s) \in \arg \max_{a} Q^{\star}_t(s,a)$. 

% Note that we can follow the standard comparison theorem for dynamic programming to argue that:
% \begin{align*}
%     V^{\star}_t(s) \geq V^{\boldsymbol{\psi}}_t(s), \forall \psi_t \in \Psi.
% \end{align*}
% Thus, if a policy satisfies the optimality dynamic programming condition, then it is optimal. Note that $\log x$ is monotone increasing. So
% \begin{align*}
%     \argmax_{\psi} \frac{1}{\lambda} \log (F(\psi)) = \argmax_{\psi} F(\psi).
% \end{align*}

% Thus, $\psi^{\star}_t(s) = \argmax_{\psi_t \in \Psi} \sum_{a} \psi_t(a \mid s) \exp(\lambda Q^{\star}(s,a))$, which is same as:
% \begin{align*}
%     \psi^{\star}_t(s) &= \argmax_{a \in \ALPHABET{A}} \exp(\lambda Q^{\star}_t(s, a)) \\
%     &= \argmax_{a \in \ALPHABET{A}} Q^{\star}_t(s, a)
% \end{align*}
% Thus, we have that:
% \begin{align*}
%     V^{\star}_t(s) = \max_{a \in \ALPHABET{A}} Q^{\star}_t(s, a).
% \end{align*}

\section{Main approach} \label{sec:risk-aware-approach}
A risk-seeking MDP involves a risk-seeking objective. We are interested in optimizing the original risk-neutral objective (expected utility) and we want to use the ideas from the risk-aware literature to converge to better local optima through using a risk-seeking objective in the start, but then gradually annealing down the risk-seeking temperature $\lambda$. Ultimately, we are left optimizing for the original risk-neutral objective. Thus, a risk-seeking objective is used to provide a better initialization for the risk-neutral approach. The final converged solution obtained is for the risk-neutral setting.

% The basic approach mentioned in Sec.~\ref{sec:basic-approach} works fine but is not a standard way to incorporate cost-based risk according to the literature \cite{shen2013risk, wang2025planning}. See \cite{gangulyrisk} for how utility-based risk can be incorporated (this is less common, but just involves flipping the signs and changing convexity to concavity).

\subsection{Applying the risk measure to the Dec-POMDP}

For simplicity of notation, we fix $N=2$. However, the results can easily be extended to any general number of agents $N$. Recall that $y_t \coloneqq (y^1_t, y^2_t)$, $z_t\coloneqq (z^1_t, z^2_t)$ and $a_t \coloneqq (a^1_t, a^2_t)$. We will use the notation $\pi_t(a_t, z_t \mid y_t, z_{t-1}) \coloneqq \pi^1_t(a^1_t, z^1_t \mid y^1_t, z^1_{t-1}) \pi^2_t(a^2_t, z^2_t \mid y^2_t, z^2_{t-1})$. 

\textbf{Marginal distribution induced by a policy:} 
The joint environment and agent state process (for all agents) $\{(s_t, y_t, z_{t-1})\}_{t \ge 1}$ is a controlled Markov chain for a given policy $\boldsymbol{\pi}$. 
Let 
\[
{\boldsymbol \zeta}^{\boldsymbol \pi} = (\zeta_1, \zeta^{\cev{\pi}_2}_2, \zeta^{\cev{\pi}_3}_3, \cdots, \zeta^{\cev{\pi}_T}_T)
\]
denote the marginal distribution on $(s_t, y_t, z_{t-1})$ induced by the policy $\boldsymbol{\pi}$, i.e.,
\begin{multline*}
\zeta^{\cev{\pi}_t}_t(s,y,z_{-}) = \PR^{\cev{\pi}_t}(s_t, y_t, z_{t-1} = s,y,z_{-} \mid \zeta_1, \cev{\pi}_t).
\end{multline*}
For ease of notation, we continue to use $\zeta_t^{\cev{\pi}_t}$ to denote the marginal and conditional distributions w.r.t.\ $\zeta_t^{\cev{\pi}_t}$. In particular, for marginals we use $ \zeta_t^{\cev{\pi}_t}(y, z_{-})$ to denote $\sum_{s \in \ALPHABET S} \zeta_t^{\cev{\pi}_t}(s,y,z_{-})$ and  for conditionals we use $\zeta_t^{\cev{\pi}_t}(s | y, z_{-})$ to denote $\zeta_t^{\cev{\pi}_t}(s, y, z_{-}) / \zeta_t^{\cev{\pi}_t} (y, z_{-})$. 

The marginal distribution $\zeta^{\cev{\pi}_t}_t$ can be updated recursively from $\zeta^{\cev{\pi}_{t-1}}_{t-1}$ and $\pi_{t-1}$ as follows: 
\begin{align*}
    &\zeta^{\cev{\pi}_t}_t(s_t,y_t,z_{t-1}) 
    \\
    &= \smash{\smashoperator{\sum_{\left(\substack{s_{t-1},y_{t-1},\\ a_{t-1},z_{t-2}}\right)}}} 
    \zeta^{\cev{\pi}_{t-1}}_{t-1}(s_{t-1},y_{t-1},z_{t-2})
    % \\
    % & \hskip 4em \times 
    P(s_t,y_t \mid s_{t-1},a_{t-1}) 
    \\
    & \hskip 7em \times 
    \pi_{t-1}(a_{t-1} z_{t-1} \mid y_{t-1}, z_{t-2})
    %&\quad \pi^1_{t-1}(a^{1'},z^1_{-} \mid y^{1'}, z^{1'}) \pi^2_{t-1}(a^{2'},z^2_{-} \mid y^{2'}, z^{2'}) .
\end{align*}

\textbf{Applying the risk map function to the original Dec-POMDP objective:}
We define the risk-seeking Dec-POMDP objective by applying the risk map at each time step in a dynamic programming fashion, where the probabilities are obtained from the marginal distribution $\{\zeta^{\cev \pi_t}_t\}_{t=1}^T$ for a given policy $\boldsymbol{\pi}$.

The risk-seeking objective for a time-varying policies $\boldsymbol{\pi^1} = (\pi^1_1, \cdots, \pi^1_T)$ and $\boldsymbol{\pi^2} = (\pi^2_1, \cdots, \pi^2_T)$ with marginal distribution $\{\zeta^{\cev \pi_t}_t\}_{t=1}^T$ can be written as follows:
\begin{align}
    J^{\risk}_{1:T}(\boldsymbol{\pi}) \coloneqq %\sum \zeta_1(s_1, y_1) \pi^1_1(a^1_1 \mid y^1_1, z^1_0) \pi^2_1(a^2_1 \mid y^2_1, z^2_0) r(s_1, a_1) + \nonumber \\
    \risk^{\pi_1}_{\zeta_1}( r(S_1, A_1) + &\risk^{\pi_{2}}( r(S_2, A_2) + \cdots \nonumber \\
    &  + \risk^{\pi_{T}} ( r(S_T, A_T) ) ) ).
\end{align}

We explain how to compute this risk-seeking objective with the help of a centralized Q-function in the next section.

% The goal is to solve for the optimal time-varying policy for the risk-seeking Dec-POMDP objective:
% \begin{align*}
%     \boldsymbol{\pi^{\star}} \coloneqq \sup_{\boldsymbol{\pi} \in \boldsymbol{\Pi}} J^{\risk}(\boldsymbol{\pi}).
% \end{align*}

\subsection{Risk-seeking conservative policy iteration (RS-CPI)}

We now present the CPI approach modified with the influence of risk-seeking rewards. In particular, we want to optimize the original risk-neutral objective, but in order to improve the local optimum that we converge to, we propose to use a risk-seeking objective during the CPI iterations. We start with a value of $\lambda > 0$ and reduce the value to $\lambda \to 0$ as the CPI iterations proceed to completion.

\textbf{Risk-seeking centralized Q-function.}
Consider the centralized Q-function:
\begin{align*}
    &Q^{\vec \pi_T}_{T}(s,y,z_{-},a,z) \coloneqq r(s, a)
\end{align*}
and for $t \in \{T-1, \dots, 1\}$, and with $\lambda > 0$, we have
\begin{align*}
    &Q^{\vec \pi_t}_{t}(s,y,z_{-},a,z) \coloneqq 
    %r(s, a) \quad + \\
    %& \quad \risk^{\pi_{t+1}} [ Q^{\vec \pi_{t+1}}_{t+1}(S_{+},Y_{+},z, \pi^1_{t+1}(Y^1_{+},z^1), \pi^2_{t+1}(Y^2_{+},z^2)) ]\\
    % &= r(s, a^1, a^2) + \smash{\sum_{s_{+},y^1_{+},y^2_{+}}} P(s_{+},y^1_{+},y^2_{+}|s,a^1,a^2) \\
    % &Q^{\vec \pi_{t+1}}_{t+1}(s_{+},y^1_{+},y^2_{+},z^1,z^2, \pi^1_{t+1}(y^1_{+},z^1), \pi^2_{t+1}(y^2_{+},z^2)). \\
     r(s, a) +  \\
    & \quad \frac{1}{\lambda} \log \Big\{ \sum_{\left(\substack{s_{+},y_{+}, \\ a_{+}, z_{+}}\right)} P(s_{+},y_{+}|s,a) \pi_{t+1}(a_{+}, z_{+} \mid y_{+},z) \\
    & \quad \exp(\lambda Q^{\vec \pi_{t+1}}_{t+1}(s_{+},y_{+},z, a_{+}, z_{+})) ) \Big\}.
\end{align*}

% For the set of team policies $\boldsymbol{\pi}$, we can write the policy evaluation $Q$-functions as follows:
% \begin{align}
% Q^{\vec \pi_T}_{T}(s,y^1,y^2,z^1_{-},z^2_{-},a^1,a^2,z^1,z^2) = r(s,a^1,a^2)
% \end{align}
% and for $t \in \{T-1, \dots, 1\}$, we have
% \begin{align}
%     &Q^{\vec \pi_t}_t(s,y^1,y^2,z^1_{-},z^2_{-},a^1,a^2,z^1,z^2) = r(s,a^1,a^2) \nonumber \\
%     & \quad\quad + \smash{\sum_{s_{+},y^1_{+},y^2_{+}}} P(s_{+},y^1_{+},y^2_{+}|s,a^1,a^2) \nonumber \\
%     &Q^{\vec \pi_{t+1}}_{t+1}(s_{+},y^1_{+},y^2_{+},z^1,z^2, \pi^1_{t+1}(y^1_{+},z^1), \pi^2_{t+1}(y^2_{+},z^2)).
% \end{align}

% \subsubsection{Optimizing for each agent}

\textbf{A policy update procedure.}
We now present a procedure for updating a policy. The procedure consists of the following steps:
\looseness=-1
\begin{enumerate}
    \item Given a policy $\boldsymbol{\pi}$, compute the risk-seeking $Q$-functions $\{Q^{\vec{\pi}_t}_t\}_{t=1}^T$ and compute the marginal distribution $\{\zeta^{\cev \pi_t}_t\}_{t=1}^T$.
    \item Compute the sequence of \emph{risk-seeking averaged local $Q$-functions} defined as follows for a given agent $i$ at time $t$:
    \begin{align*}
        \hskip 1em & \hskip -1em
        \bar Q^{\cev \pi_t, \pi^{-i}_t, \vec \pi_t}_t(y^i, z^i_{-}, a^i, z^i) \\
        % &= \risk_{\zeta^{\cev \pi_t, \pi^{-i}_t}_t} \left[ Q^{\vec \pi_t}_t(S,Y,Z_{-},A,Z) \bigm| y^i, z^i_{-}, a^i, z^i\right] \\
        % &= \frac{1}{\lambda} \log \Big\{ \sum_{\scalebox{0.6}{$s, y^{-i}, z^{-i}_{-}, a^{-i}, z^{-i}$}} \zeta^{\cev \pi_t, \pi^{-i}_t}_t(s, y^{-i}, z^{-i}_{-}, a^{-i}, z^{-i} \mid y^i, z^i_{-}) \\
        &= \frac{1}{\lambda} \log \Big\{ \sum \zeta^{\cev \pi_t, \pi^{-i}_t}_t( \cdot \mid y^i, z^i_{-}) \\
        &\quad \exp( \lambda Q^{\vec \pi_t}_t(s,y,z_{-},a,z)) \Big\}.
    \end{align*}
    where the summation is over $s, y^{-i}, z^{-i}_{-}, a^{-i}, z^{-i}$.
    \item Compute an updated policy $\bar \pi^{i}_t$ as follows:
    \begin{equation}\label{eq:pi-update}
        \bar \pi^{i}_t(y^i,z^i_{-}) \in \argmax_{(a^i,z^i)} %\in \ALPHABET A \times \ALPHABET Z}
        \bar Q^{\cev \pi_t, \pi^{-i}_t, \vec \pi_t}_t(y^i, z^i_{-}, a^i, z^i),
    \end{equation}
    and finally replace $\pi^{i}_t$ with $\bar \pi^{i}_t$.
\end{enumerate}
We will use the short-hand notation $\bar \pi^i_t \in \GREEDY^i_t(\boldsymbol{\pi})$ to denote the update procedure for agent $i$ at time $t$. A key feature of the policy update procedure is the following.
\looseness=-1
\begin{proposition}[Performance Improvement Guarantee]\label{prop:improvement}
    For any $\boldsymbol{\pi} \in \Pi$, any $t$ and any agent $i$, let $\bar \pi^i_t \in \GREEDY^i_t(\boldsymbol{\pi})$. Then, we have
    $
    J^{\risk}_{t:T}(\cev \pi_t, \bar \pi^i_t, \pi^{-i}_t, \vec \pi_t)
    \ge
    J^{\risk}_{t:T}(\cev \pi_t, \pi_t, \vec \pi_t).
    $
    
\end{proposition}
\begin{proof}
The proof is given in Appendix~\ref{app:main-proof}.
\end{proof}

\textbf{Conservative Policy Update.} \label{subsubsec:cpi-like-update}
Following \citet{kakade2002approximately}, we present a conservative policy update procedure. Given a forgetting factor $\alpha \in [0,1]$, define the conservative policy update operator $\GREEDY^{i, \alpha}_t$ as follows:
\[
    \pi^{i, \alpha}_t \in \GREEDY^{i, \alpha}_t(\boldsymbol{\pi}) \coloneqq (1-\alpha) \pi^i_t + \alpha \GREEDY^i_t(\boldsymbol{\pi}).
\]
Thus, for $\alpha=1$, the conservative policy update is same as $\GREEDY^i_t(\boldsymbol{\pi})$, but for $\alpha < 1$, the conservative policy update mixes between the old policy $\pi_t$ and a policy in $\GREEDY^i_t(\boldsymbol{\pi})$.

% \red{
% Theorem similar to POCPI paper showing monotonic performance improvements
% }

\begin{corollary}\label{cor:improvement}
    The result of Proposition~\ref{prop:improvement} holds for $\bar \pi^i_t = \GREEDY^{i, \alpha}_t(\boldsymbol{\pi})$ for any $\alpha \in [0,1]$. 
\end{corollary}
\begin{proof}
    % Due to linearity of expectations, we have 
    % \[
    %     J^{(\cev \pi_t, \pi^{-i}_t, \pi^{i, \alpha}_t, \vec \pi_t)} 
    %     = 
    %     (1 - \alpha) J^{(\cev \pi_t, \pi^{-i}_t, \bar \pi^i_t, \vec \pi_t)} 
    %     + 
    %     \alpha J^{\GREEDY^{i, \alpha}_t(\boldsymbol{\pi})}
    % \]
    % Therefore, the result follows from Proposition~\ref{prop:improvement}.
    The proof is given in Appendix~\ref{app:cor-proof}.
\end{proof}

\textbf{Risk-seeking algorithm.}
We combine the risk-seeking ideas mentioned earlier with the conservative policy update to design the Risk-Seeking Conservative Policy Iteration (RS-CPI) algorithm described in Algorithm~\ref{algo:policy-improvement-risk-seeking}. This algorithm requires a sequence of decreasing temperatures $\{ \lambda^{(k)} \}_{k=1}^K$ and forgetting rates $\{ \alpha^{(k)} \}_{k=1}^K$.

\begin{algorithm}[!ht]
\caption{RS-CPI for Dec-POMDPs}\label{algo:policy-improvement-risk-seeking}
\textbf{Input}: $\boldsymbol{\pi}^{(0)} = (\pi^{(0)}_1, \cdots, \pi^{(0)}_T)$, $K$ steps\\
\textbf{Input}: Forgetting rates $(\alpha^{(1)}, \alpha^{(2)}, \cdots, \alpha^{(K)})$.\\
\textbf{Input}: Temperature rates $(\lambda^{(1)}, \lambda^{(2)}, \cdots, \lambda^{(K)})$.
\begin{algorithmic}[1]
\STATE $k \gets 0$
\WHILE{$k < K$} \label{algo:DP-outer-loop-risk-seeking}
    \STATE $k \gets k+1$
    \STATE $\lambda \gets \lambda^{(k)}$
    \FOR{$t\gets T:1$} \label{algo:DP-inner-loop-risk-seeking}
        \STATE $\pi^{(k)}_t \gets \GREEDY^{\alpha^{(k)}}_t(\cev{\pi}^{(k-1)}_t, \pi^{(k-1)}_t, \vec{\pi}^{(k)}_t)$
        $\vphantom{\Big(}$
        \label{algo:DP-single-update}
    \ENDFOR
\ENDWHILE
\STATE \textbf{return} $\boldsymbol{\pi}^{(k)}$
\end{algorithmic}
\end{algorithm}

% \red{Theorem saying that algo 1 gives a local optima / Nash}

\begin{proposition}[Convergence Guarantee]\label{prop:convergence}
For any choice of forgetting rate $\boldsymbol{\alpha}$ and fixed $\lambda$, the performance $\{ J^{\risk}_{1:T}(\boldsymbol{\pi}^{(k)})\}_{k \ge 0}$ of the policies generated by RS-CPI converges to a limit $\boldsymbol\pi^\circ = (\pi^\circ_1, \dots, \pi^\circ_T)$ that has the following local optimality property:
\begin{equation}\label{eq:local-optimality}
    \pi^{i,\circ}_t \in \GREEDY^i_t(\boldsymbol{\pi^{\circ}}), \quad \forall t \in \{1, \dots, T\}, \forall i.
\end{equation}
When conservatism is absent (i.e., $\alpha^{(k)} = 1$ for all $k$), 
the convergence takes place in a finite number of steps.
\end{proposition}
\begin{proof}
    The proof is given in Appendix~\ref{app:prop-convg-proof}.
\end{proof}

\section{Numerical Experiments}
In this section, we present a detailed numerical study to compare the performance of the proposed RS-CPI algorithm with multiple state-of-the-art planning baselines for Dec-POMDPs on multiple benchmark environments.

\subsection{Details of the experiments}

\textbf{Environments.}
% \red{state why these specific environments were chosen}
We use the Dec-POMDPs from the MASPlan benchmark \cite{masplan}, which is a well known standard for evaluating Dec-POMDPs. We consider two small environments and the two largest environments from the MASPlan benchmark to show that our approach generalizes well across domains. Particularly, we use the size of the history-based policy space (given in Table~$3$ of Sec.~$5.1$, p$40$  from \citet{dibangoye2016optimally}) as a measure for the degree of difficulty in solving the Dec-POMDP, and thus, we select the following Dec-POMDPs based on the size of policy spaces:
\begin{enumerate}[noitemsep, topsep=0pt, leftmargin=1.25em]
    \item \textsc{Dec Tiger} \cite{nair2003taming}, in which two agents must coordinate their actions (opening a door) based on individual noisy observations.
    \item \textsc{Recycling Robots} \cite{amato2012optimizing}, in which two robots must independently decide whether to recycle small or large cans or return to a charging station.
    \item \textsc{Cooperative Box Pushing} \cite{seuken2012improved}, in which two agents may either push small boxes or coordinate their actions to push a bigger box.
    \item \textsc{Mars Rovers} \cite{amato2009achieving}, in which multiple rovers must coordinate their trajectories to visit various sites of interest and collect data, under constraints of limited communication.
\end{enumerate}

\textbf{Baselines.} We consider the following baselines:\footnote{We do not consider the use of the designer's approach \cite{MahajanPhD} or the common information approach \cite{nayyar2013decentralized, tang2023novel} as baselines as they do not provide implementations of their methods and the results in their papers are provided only for small Dec-POMDPs with small horizons, while we require comparisons for larger settings.}
\begin{enumerate}[noitemsep, topsep=0pt, leftmargin=1.25em]
    \item Feature-based heuristic search value iteration (FB-HSVI) \cite{dibangoye2016optimally}, which is a state-of-the-art algorithm to find $\epsilon$-optimal \emph{history-based} policies for Dec-POMDPs. It provides near-optimal policies for smaller environments, but is challenging to run for larger environments for large horizons.
    \item Policy-finding multi-agent A$^\ast$ (PF-MAA$^{\ast}$) \cite{koops2024approximate}, which is a computationally efficient heuristic algorithm to find good \emph{history-based} policies in Dec-POMDPs. While PF-MAA$^{\ast}$ lacks optimality guarantees, it provides a computable upper bound to estimate the quality of the solution. 
\end{enumerate}

% We compare the performance of our RS-CPI algorithm with the state-of-the-art $\epsilon$-approximate optimal FB-HSVI \cite{dibangoye2016optimally} algorithm and the more computationally efficient PF-MAA$^{\ast}$\cite{koops2024algorithms, koops2024approximate} algorithm.
% Although PF-MAA$^{\ast}$ is more computationally efficient it does not guarantee that the policy obtained is $\epsilon$-approximate optimal, however, they also provide tight upper-bounds for the Dec-POMDPs under consideration which effectively show that the PF-MAA$^{\ast}$ solutions obtained are very close to being optimal.

% We choose FB-HSVI \cite{dibangoye2016optimally} as a baseline to compare with because it provides near-optimal solutions for smaller problems. However, it is challenging to run this for larger problems. For larger problems we consider PF-MAA$^{\ast}$ \cite{koops2024algorithms, koops2024approximate} as a baseline. Although PF-MAA$^{\ast}$ does not guarantee optimality, \citet{koops2024algorithms, koops2024approximate} provide tight upper-bounds for the Dec-POMDPs under consideration which effectively show that the PF-MAA$^{\ast}$ solutions obtained are very close to being optimal. Another benefit of using PF-MAA$^{\ast}$ is that it can scale well to larger problems.

\textbf{Details for RS-CPI.}
For RS-CPI, we keep a constant $\alpha$ and start from an initial $\lambda > 0$ (thus, a risk-seeking objective) which is then linearly decayed to $0$, reverting to a risk-neutral objective. We then run additional iterations at $\lambda = 0$ until convergence. The values of the constant $\alpha$ and the initial $\lambda$ are independently chosen for each environment by a hyperparameter search (where performance is estimated over $5$ random seeds). Performance can be calculated by the method described in Appendix~\ref{app:joint_policy_eval}.

The initial policy is randomized.  For each constant $\alpha$ and initial $\lambda$, we re-run RS-CPI across $5$ initializations and report the best performance obtained after convergence.

% In order to select the initial $\lambda$ and the constant conservatism $\alpha$ we try $2-3$ different values for each Dec-POMDP and we run $5$ random seeds for each setting. We do this for each horizon reported in the results and select the maximum performance over all the runs. With the given starting value of $\lambda$ we decrease it towards $0$ with a linear schedule (which then becomes the original risk-neutral objective) after which we run a few more iterations with $\lambda=0$. 

\subsection{Experimental Results}

\textbf{Results for $\lvert \ALPHABET{Z}^i \rvert = 2$.}
For our initial set of experiments, we fix $\lvert \ALPHABET{Z}^i \rvert = 2$ and report the results in Table~\ref{tab:results-comparison}. We directly use the performance metrics reported by the respective authors for the baselines. We can see from the comparisons across all environments and all horizons that the performance of our RS-CPI algorithm is comparable despite using only $2$ possible memory states. Since most of the performance metrics reported from using FB-HSVI and PF-MAA$^{\ast}$ are already near-optimal (or optimal), it is not possible to show significant improvements over those baselines.

\begin{table}[!t]
\caption{Performance comparison across benchmarks and planning horizons T. RS-CPI uses $\lvert \ALPHABET{Z}^i \rvert = 2$. The ``-'' entries were not not reported in the corresponding paper.}
\label{tab:results-comparison}
\centering
\small
\begin{tabular}{c S[table-format=4.2] S[table-format=4.2] S[table-format=4.2]}
\toprule
{$T$} & {FB-HSVI} & {PF-MAA$^\ast$} & {\begin{tabular}{@{}c@{}}
RS-CPI\\
\text{(ours)}
\end{tabular}} \\
\toprule
\multicolumn{4}{c}{\textsc{Dec Tiger}} \\ \midrule
6 & 10.38 & 10.38 & 10.38 \\ %\hline
7 & 9.99 & 9.99 &  8.75\\ %\hline
8  & 12.22  & 8.36  & 8.38 \\ %\hline
9 & 15.57 & 15.57 &  15.57\\ %\hline
10  & 15.18  & 15.18  & 13.76\\ %\hline
20  & 28.75  & 29.12  & 27.63 \\ %\hline
50 & 80.66 & 81.03 &  54.11\\ %\hline
100 & 170.90 & 170.91 & 120.30 \\ %\hline
\midrule
\multicolumn{4}{c}{\textsc{Recycling Robots}} \\ 
\midrule
100  & 308.78  & 308.79  & 308.79 \\ %\hline
500 & {--} & 1539.56 & 1539.17 \\ %\hline
1000  & {--}  & 3078.02  & 3077.63 \\ %\hline
2000 & {--} & 6154.94 & 6154.56\\ %\hline
\midrule
\multicolumn{4}{c}{\textsc{Cooperative Box Pushing}} \\ 
\midrule
4  & 98.59 & 98.59  & 98.17 \\ %\hline
5 & 107.72 & 107.73 & 107.64 \\ %\hline
6  & 120.67  & 120.68  & 120.27 \\ %\hline
7 & 156.42 & 155.99 & 155.01 \\ %\hline
8  & 191.22  & 191.23  & 186.04\\ %\hline
9 & 210.27 & 210.26 & 209.61 \\ %\hline
10  & 223.74  & 224.28  & 223.78 \\ %\hline
20  & 458.10  & 474.97  &  469.39\\ %\hline
50 & 1134.70 & 1209.80 &  1192.46\\ %\hline
100 & {--} & 2433.51 & 2352.18 \\ %\hline
\midrule
\multicolumn{4}{c}{\textsc{Mars Rovers}} \\ %\hline
\midrule
6  & 18.62  & 18.62  & 18.62  \\ %\hline
7 & 20.90 & 20.90 & 20.90 \\ %\hline
8  & 22.47 & 22.48  & 22.48 \\ %\hline
9 & 24.31 & 24.32 & 24.32  \\ %\hline
10  & 26.31  & 26.31  & 26.32 \\ %\hline
20  & 52.13  & 49.37  & 50.47  \\ %\hline
50 & 128.95 & 122.56 & 126.43  \\ %\hline
100 & 249.92 & 234.06 & 252.63 \\ %\hline
\bottomrule
\end{tabular}
\end{table}

\textbf{Ablation study with risk-seeking behavior and conservatism.} We next consider the effect of adding risk-seeking and conservatism in the overall performance in Table~\ref{tab:ablation-experiments-risk-seeking-conservatism}. We perform the experiments under the same hyperparameter settings (i)~with CPI only, (ii)~with risk-seeking only, (iii)~without CPI and risk-seeking and compare it with our resulting algorithm RS-CPI that combines all of them. Clearly the ablations show that it is critical to jointly use CPI and risk-seeking together as we can see that performance deteriorates when we remove either/both of them.

\begin{table}[t]
\caption{Ablations with risk-seeking and conservatism for the different planning horizons $T$. All experiments use $\lvert \ALPHABET{Z}^i \rvert = 2$.}
\label{tab:ablation-experiments-risk-seeking-conservatism}
\centering
\small
\begin{tabular}{c S[table-format=4.2] S[table-format=4.2] S[table-format=4.2] S[table-format=4.2]}
\toprule
{$T$} & {CPI only} & {RS only} & {No CPI + No RS} & {RS-CPI}
\\ 
\toprule
\multicolumn{5}{c}{\textsc{Dec Tiger}} \\ \midrule
6   & 10.38 & -12.00 & -12.00  & 10.38 \\ %\hline
7   & 8.75 & -14.00 & -14.00 & 8.75 \\ %\hline
8   & 6.50 & -16.00 & -16.00 & 8.38 \\ %\hline
9   & 10.10 & -18.00 & -18.00 & 15.57 \\ %\hline
10  & 8.42 & -20.00 & -20.00 & 13.76  \\ %\hline
20  & 23.17 & -40.00 & -40.00 & 27.63  \\ %\hline
50  & 47.95 & -100.00 & -100.00 & 54.11  \\ %\hline
100 & 91.93 & -200.00 & -200.00 & 120.30  \\ %\hline
\toprule
\multicolumn{5}{c}{\textsc{Recycling Robots}} \\ \midrule
100  & 308.40 & 308.40 & 308.40  & 308.79  \\ %\hline
500  & 1536.21 & 1536.21 & 1536.21 & 1539.17  \\ %\hline
1000 & 3074.56 & 3071.01 & 3071.01 & 3077.63  \\ %\hline
2000 & 6148.05 & 6143.33 & 6143.33 & 6154.56  \\ %\hline
\toprule
\multicolumn{5}{c}{\textsc{Cooperative Box Pushing}} \\ \midrule
4   & 19.43 & 18.70 & 19.43 & 98.17  \\ %\hline
5   & 35.28 & 26.39 & 35.28 & 107.64  \\ %\hline
6   & 78.56 & 35.94 & 78.37 & 120.27  \\ %\hline
7   & 82.73 & 36.96 & 79.32 & 155.01  \\ %\hline
8   & 53.00 & 44.17 & 53.00 & 186.04  \\ %\hline
9   & 96.55 & 53.60 & 95.63 & 209.61  \\ %\hline
10  & 55.95 & 55.20 & 103.25 & 223.78  \\ %\hline
20  & 124.03 & 115.44 & 157.79 & 469.39  \\ %\hline
50  & 435.48 & 294.58 & 404.04 & 1192.46  \\ %\hline
100 & 771.35 & 597.68 & 753.44 & 2352.18  \\ %\hline
\toprule
\multicolumn{5}{c}{\textsc{Mars Rovers}} \\ \midrule
6   & 10.89 & 18.62 & 10.75 & 18.62  \\ %\hline
7   & 12.63 & 18.42 & 12.60 & 20.90  \\ %\hline
8   & 14.34 & 18.22 & 12.97 & 22.48  \\ %\hline
9   & 16.16 & 18.64 & 15.83 & 24.32  \\ %\hline
10  & 18.02 & 24.68 & 17.74 & 26.32  \\ %\hline
20  & 35.80 & 49.37 & 35.72 & 50.47  \\ %\hline
50  & 89.37 & 124.31 & 88.21 & 126.43  \\ %\hline
100 & 178.60 & 249.29 & 178.40 & 252.63  \\ %\hline
\bottomrule
\end{tabular}
\end{table}

\textbf{Ablation study with increasing agent state size $\lvert \ALPHABET{Z}^i \rvert$.}
It is natural to expect greater performance if we allow larger agent state sizes. Thus, we perform this experiment for $\lvert \ALPHABET{Z}^i \rvert = 1,2,3,4$. Note that $\lvert \ALPHABET{Z}^i \rvert = 1$ corresponds to the open-loop policy (reactive policy). We can see from the results that using more memory helps most in the case of Dec-Tiger, where the reactive policy performs poorly. However, surprisingly, the reactive policies on the Cooperative Box Pushing and Mars Rovers Dec-POMDPs perform well and adding memory to these agents does not lead to a significant performance increase. We give intuitions on the obtained reactive policies for \textsc{Cooperative Box Pushing} and \textsc{Mars Rovers} in Appendix~\ref{app:reactive-policy-descriptions} and we highlight how it is possible to obtain reactive policies that are near-optimal for these Dec-POMDPs.

This suggests that while these Dec-POMDPs are significantly larger, they are still simpler than the smaller Dec-Tiger Dec-POMDP in the sense that they can consider simplified reactive policies without having to worry about maintaining an explicit agent state representation. As a result, this greatly simplifies the policy search process in Cooperative Box Pushing and Mars Rovers Dec-POMDPs which means that apart from having large state, observation and action spaces --- these problems are not as hard as one might think. This highlights the requirement for more challenging benchmarks for Dec-POMDPs, not only in the sense of larger spaces, but also in terms of agents requiring memory to perform well.

\textbf{Polynomial space and runtime algorithm.}
The RS-CPI algorithm has polynomial runtime and memory in the size of the state, observation, action, agent-state spaces of the individual agents. While this is still large due to the product of several spaces being considered, it is still far better in terms of computational complexity than dealing with exponential runtime/memory algorithms. This makes our algorithm much more scalable for remembering events occurring over larger horizons. Note that baseline approaches do show better results in terms of time taken to obtain the optimal policy, however, they use a limited sliding window of size $3-4$ which may not work in general. We provide graphs illustrating our runtime and memory scaling over agent state size $Z$ and the horizon $T$ in Appendix~\ref{app:memory-runtime-analysis}.

\begin{table}[t]
\caption{Table for comparison of number of agent states $\lvert \ALPHABET{Z}^i \rvert$ with performance. We did not perform experiments for \textsc{Mars Rovers} for $\lvert \ALPHABET{Z}^i \rvert = 4$.}
\label{tab:results-z-ablation}
\centering
\small
\begin{tabular}{c S[table-format=4.2] S[table-format=4.2] S[table-format=4.2] S[table-format=4.2]}
\toprule
{$T$} & {$\lvert \ALPHABET{Z}^i \rvert =1$} & {$\lvert \ALPHABET{Z}^i \rvert =2$} & {$\lvert \ALPHABET{Z}^i \rvert = 3$} & {$\lvert \ALPHABET{Z}^i \rvert = 4$}  \\ 
\toprule
\multicolumn{5}{c}{\textsc{Dec Tiger}} \\ \midrule
20 & -31.77 & 27.63  & 31.58  & 30.69  \\ %\hline
50 & -80.90 & 54.11 & 71.64 & 69.71  \\ %\hline
100 & -123.01 & 120.30 & 145.14 & 152.04 \\ %\hline

\toprule
\multicolumn{5}{c}{\textsc{Cooperative Box Pushing}} \\ \midrule
20 & 448.09 & 469.39  & 472.97  & 474.25  \\ %\hline
50 & 1112.67 & 1192.46 & 1200.06  & 1204.52 \\ %\hline
100 & 2228.57 & 2352.18 & 2410.70 & 2421.33 \\ %\hline

\toprule
\multicolumn{5}{c}{\textsc{Mars Rovers}} \\ \midrule
20 & 50.47 & 52.36  & 50.59 & {--}  \\ %\hline
50 & 126.43 & 126.34 & 127.24 & {--}\\ %\hline
100 & 252.63 & 251.95  & 253.93 & {--} \\ %\hline
\bottomrule
\end{tabular}
\end{table}

\section{Conclusions}

We propose an algorithm RS-CPI that combines risk-seeking and conservatism in a novel way to find policies which operate with finite memory constraints. We show in existing Dec-POMDP benchmarks that our approach performs comparably to the state-of-the-art approaches despite using limited memory. In addition, we show ablations with the performance of our algorithm when increasing the agent state size and we observe that performance increases as more memory is allowed. Another interesting observation is that reactive policies perform better than expected in certain Dec-POMDPs which are considered hard because of their large spaces. While being larger does make a Dec-POMDP harder to solve, another important criteria to gauge the difficulty of a Dec-POMDP should be the amount of memory required by agents to perform optimally. If the agents need less memory, then the Dec-POMDP is much simpler to solve, because the search space of policies is much smaller and simpler.

\bibliography{references}
\bibliographystyle{icml2026}

%%%%%%%%%%%%%%%%%%%%%%%%%%%%%%%%%%%%%%%%%%%%%%%%%%%%%%%%%%%%%%%%%%%%%%%%%%%%%%%
%%%%%%%%%%%%%%%%%%%%%%%%%%%%%%%%%%%%%%%%%%%%%%%%%%%%%%%%%%%%%%%%%%%%%%%%%%%%%%%
% APPENDIX
%%%%%%%%%%%%%%%%%%%%%%%%%%%%%%%%%%%%%%%%%%%%%%%%%%%%%%%%%%%%%%%%%%%%%%%%%%%%%%%
%%%%%%%%%%%%%%%%%%%%%%%%%%%%%%%%%%%%%%%%%%%%%%%%%%%%%%%%%%%%%%%%%%%%%%%%%%%%%%%
\newpage
\appendix
\onecolumn
\section{Exact policy evaluation for agent-state based polices in polynomial time} \label{app:joint_policy_eval}

In this section, we show how policy evaluation can be done in a Dec-POMDP for a set of joint agent-state based policies. 
The key intuition is Lemma~\ref{lem:Markov}.
\begin{lemma}\label{lem:Markov}
    For any joint agent-state based policy $\boldsymbol{\pi}$, the process $\{(S_t, Z_t)\}_{t \ge 1}$ is Markov. The processes $\{(S_t, Z_t, A_t)\}_{t \ge 1}$, $\{(S_t, Y_t, Z_t)\}_{t \ge 1}$ and $\{(S_t, Y_t, Z_t, A_t)\}_{t \ge 1}$ are also Markov.
\end{lemma}
Thus, there exists a centralized MDP which has the following time-varying reward-function and dynamics:
\begin{align}
    \bar r^{\pi_t}_t(s,y,z) &= \sum_{a \in \ALPHABET A} \prod_{i=1}^{N} \pi^i_t(a^i|y^i, z^i) r(s,a),  \\
    \bar P^{\pi_t}_t(s',y',z'|s,y,z) &= \sum_{a \in \ALPHABET A} \prod_{i=1}^{N} \pi^i_t(a^i, {z^i}'| y^i, z^i) P(s',y'| s,a) %\IND_{\{{z^i}' = \pi(z^i,{y^i}')\}}.
\end{align}

We can then evaluate the performance of this time-varying Markov chain via the performance evaluation formulas for MDPs. In particular, define
\begin{align*}
    \tilde V(S_0,Y_0,Z_0) &= \bar r^{\pi_1}_1 + \bar P^{\pi_1}_1 \bar r^{\pi_2}_2 + \dots + \bar P^{\pi_1}_1 \bar P^{\pi_2}_2 \cdots \bar P^{\pi_{T-1}}_{T-1} \bar r^{\pi_T}_T.
\end{align*}
to be the initial value function, where $Z_0$ is arbitrary and $(S_0,Y_0)$ are governed by the initial distribution $\zeta_1$. Thus the final performance can be calculated as:
\begin{align*}
    J(\pi) = \sum_{s,y} \zeta_1(s,y) \tilde V(s,y,Z_0).
\end{align*}
Also note that we are only reporting final performance for the risk-neutral problem, however, if one is interested in the risk-seeking problem's performance, then the expectations should be modified by the appropriate function, e.g.: logsumexp function in the case of entropic risk.

\section{Risk-aware MDPs with rewards instead of costs based on \cite{shen2013risk, wang2025planning}}

\subsection{Risk measure}

First consider the definition for the basic risk map from \citet{shapiro2021lectures}, which consists of a mapping from a random variable (reward) to a real number and essentially quantifies the risk of that particular random variable. Consider a finite probability space $(\Omega, \mathcal{F}, \mathcal{P})$: where $\Omega$ is a finite set of outcomes, $\mathcal{F}$ is the sigma-algebra of the outcomes, and $\mathcal{P}$ is the associated probability mass function. Let $\mathcal{L}(\Omega)$ denote the set of finite real-valued functions (random variables) on $\Omega$.
\begin{definition}
Given a finite probability space $(\Omega, \mathcal{F}, \mathcal{P})$, a risk measure $\rho \colon \mathcal{L}(\Omega) \to \reals$ maps a random variable (reward) to a real value quantifying its risk. The risk measure $\rho$ is called a monetary risk measure if it satisfies the following three properties:
\begin{enumerate}
    \item \textbf{Monotonicity}: $\rho(v) \leq \rho(w)$, for all $v,w \in \mathcal(\Omega)$ such that $v \leq w$ a.s.
    \item \textbf{Translation Invariance}: $\rho(v + \lambda) = \rho(v) + \lambda$ for any $v \in \mathcal{L}(\Omega), \lambda \in \reals$.
    \item \textbf{Normalization}: $\rho(0) = 0$.
\end{enumerate}
Furthermore, such a monetary risk measure $\rho$ may also satisfy the following properties:
\begin{enumerate}
    \item \textbf{Convexity}: For all $v,w \in \mathcal{L}(\Omega), \alpha \in [0, 1]$, $\rho(\alpha v + (1-\alpha)w) \leq \alpha \rho(v) + (1-\alpha)\rho(w)$.
    \item \textbf{Concavity}: For all $v,w \in \mathcal{L}(\Omega), \alpha \in [0, 1]$, $\rho(\alpha v + (1-\alpha)w) \geq \alpha \rho(v) + (1-\alpha)\rho(w)$.
    \item \textbf{Positive Homogeneity}: For all $v \in \mathcal{L}(\Omega), \lambda \geq 0$, $\rho(\lambda v) = \lambda \rho(v)$.
\end{enumerate}
If the monetary risk measure $\rho$ satisfies concavity and positive homogeneity properties, then it is called a coherent risk measure.
\end{definition}

\begin{remark}
    Since we are using the reward formulation instead of the standard cost formulation, a higher value given by $\rho$ indicates a more desirable outcome. So ideally we would want to maximize the objective constructed using the risk measure $\rho$.
\end{remark}

Coherent risk measures are common in the finance literature. Concavity is associated with risk-averse behavior and convexity is associated with risk-seeking behavior. See Section~$4.2$ of \citet{shen2013risk} for an intuitive explanation of how concave/convex corresponds to risk-averse/seeking in terms of economics and uncertainty.

\begin{remark}
Note that outside being a monetary risk measure, none of the more advanced properties like concavity and coherency are necessary to consider risk-aware behavior. Although it is convenient to consider concave/convex risk maps to guarantee the appropriate risk-aware behavior everywhere.
\end{remark}

\subsection{Markov decision processes}

A Markov decision process (MDP) models sequential decision-making problems in which a single agent interacts with a stochastic environment over time. A finite-horizon MDP is defined by the tuple
\[
\mathcal{M} = \langle \ALPHABET{S}, \ALPHABET{A}, P, r, \zeta_1, T \rangle,
\]
where:
\begin{itemize}
    \item $\ALPHABET{S}$ is a finite set of environment states.
    \item $\ALPHABET{A}$ is a finite set of actions available to the agent.
    \item $P$ is the state transition kernel, where
    \[
    P(s' \mid s, a) = \mathbb{P}(s_{t+1} = s' \mid s_t = s,\, a_t = a),
    \]
    for all $s,s' \in \ALPHABET{S}$ and $a \in \ALPHABET{A}$.
    \item $r : \ALPHABET{S} \times \ALPHABET{A} \rightarrow \mathbb{R}$ is the reward function.
    \item $\zeta_1 \in \Delta(\ALPHABET{S})$ is the initial state distribution.
    \item $T \in \mathbb{N}$ is the finite decision horizon.
\end{itemize}

At each time step $t = 1,\cdots,T$, the system is in state $s_t \in \ALPHABET{S}$.
The agent selects an action $a_t \in \ALPHABET{A}$, receives a reward $r(s_t, a_t)$,
and the environment transitions to a new state $s_{t+1}$ according to the transition
kernel $P$. A decision making rule (or policy) at each time step can be used to select the action $a_t$ based on the state $s_t$, i.e., $\psi_t \colon \ALPHABET{S} \to \Delta(\ALPHABET{A})$. We refer to the combined decision making rules over all time jointly as the overall policy $\boldsymbol{\psi} = (\psi_1, \cdots, \psi_T)$.

The objective function $J$ for an MDP is the cumulative reward over the finite horizon $T$ and is a function of the policy:
\begin{align*}
    J(\boldsymbol{\psi}) = \EXP^{\boldsymbol{\psi}} \left[ \sum_{t=1}^{T} r(S_t, A_t) \right]
\end{align*}

The objective is to find a policy $\boldsymbol{\psi}$ that maximizes the cumulated reward over time $J(\boldsymbol{\psi})$:
\begin{align*}
    \boldsymbol{\psi^{\star}} \coloneqq \sup_{\boldsymbol{\psi} \in \boldsymbol{\Psi}} J(\boldsymbol{\psi}).
\end{align*}

\subsection{Risk maps}

Given the definition of a risk measure, we can now look at the definition of a risk map $\risk$ applied to an MDP $\mathcal{M}$ as follows.

\begin{definition}
    A risk map $\risk$ is a function that maps each state action pair $(s,a) \in \ALPHABET S \times \ALPHABET A$ to a monetary risk measure on the space $(\ALPHABET S, P(\cdot \mid s, a))$. For a given value function random variable $V$ (over all states), the risk map for $(s,a)$ is given by $\risk(V \mid s,a)$. Suppose we are using a policy $\psi$ at state $s$, then we write this as $\risk^{\psi}(V | s)$. This can be further simplified to $\risk^{\psi}_s(V)$.
\end{definition}

The risk aware objective for a time-varying policy $\boldsymbol{\pi} = (\psi_1, \cdots, \psi_T)$ can then be written as follows:
\begin{align}
    J^{\risk}(\boldsymbol{\psi}) &\coloneqq r^{\psi_1}(S_1) +  \risk^{\psi_1}_{S_1}( r^{\psi_2}(S_2) + \risk^{\psi^{2}}_{S_2}( r^{\psi_3}(S_3) + \cdots + \risk^{\psi^{T-1}}_{S_{T-1}} (r^{\psi_T}(S_T) ) ) ).
\end{align}

The goal is to solve for the optimal time-varying policy for the risk-aware MDP:
\begin{align*}
    \boldsymbol{\psi^{\star}} \coloneqq \sup_{\boldsymbol{\psi} \in \boldsymbol{\Psi}} J^{\risk}(\boldsymbol{\psi}).
\end{align*}

\subsection{Example choices for risk map functions}

The following are some examples from \citet{shen2013risk} for risk map functions $\risk$.
\begin{enumerate}
    \item \textbf{Classical MDP risk neutral}: This is the standard risk-neutral reward function used in MDPs. Since it is linear in the value function $V$, it satisfies concavity and convexity. Since it also satisfies positive homogeneity, it is coherent as well.
    \begin{align*}
        \risk(V \mid s,a) &\coloneqq \EXP^{P} \left[ V \mid s, a \right] = \sum_{s'} P(s' \mid s, a) V(s').
    \end{align*}
    \item \textbf{Entropic map}: This is based on the logsumexp function. The parameter $\lambda \in \reals$ controls the risk-preference of $\risk$. If $\lambda > 0$, then $\risk$ is convex everywhere and therefore risk-seeking everywhere. If $\lambda < 0$, $\risk$ is concave everywhere and therefore risk-averse everywhere. It can also be shown that the limit as $\lambda \to 0$ is the same as the classical MDP risk neutral reward.
    \begin{align*}
        \risk(V \mid s,a) &\coloneqq \frac{1}{\lambda} \log \EXP^{P} \left[ \exp (\lambda V) \mid s, a \right] \\
        &= \frac{1}{\lambda} \log \left\{ \sum_{s'} P(s' \mid s, a) \exp(\lambda V(s')) \right\}.
    \end{align*}
    \item \textbf{Mean-semideviation trade-off}: This involves modifying risk based on deviation from the expected values of rewards. $\lambda \in \reals$ works similar to the entropic map. The notation $x_+$ denotes $\max(x, 0)$ and $k \geq 1$.
    \begin{align*}
        \risk(V \mid s,a) &\coloneqq \EXP^{P} \left[ V \mid s, a \right] + \lambda \EXP^{P} \left[ (V - \EXP^{P}[V \mid s, a] )^k_+  \mid s, a \right]^{\frac{1}{k}}.
    \end{align*}
\end{enumerate}

\section{Related work} \label{app:related-work}

\subsection{Monotonic improvement dynamic programs / policy iteration schemes for POMDPs}
A similar approach to the monotonic improvement DP explored in this paper was developed by \citet{van2025memoryless}, which applies the same idea but for a single agent POMDP to learn a time-varying reactive (considers only the most recent observation as input) policy for the finite horizon case.

\subsection{Risk-aware methods}

Instead of using risk-neutral MDPs, we can consider risk-aware MDPs as \citet{wang2025planning}. This involves modifying the original expected reward to a reward which is mapped through some risk measure. Such risk measures can lead to risk-averse/risk-seeking behaviors depending on the hyper-parameters used.

\citet{shen2013risk} give a very useful overview of general risk measures and how to appropriately apply them for MDPs (since most papers in the space use a cost-based risk-averse approach). They define a risk measure and explain the specific of the properties that a risk measure needs in order to be considered risk-averse or risk-seeking for the MDP case. They look at obtaining time-homogeneous policies for the infinite horizon discounted case and the infinite horizon average cost setting. They show the existence of optimal policies for both these settings and show that it is well defined for the MDP case. While \citet{wang2025planning} also consider time-homogeneous policies in the average cost setting for risk-aware MDPs, the distinction in their work is that they present planning and learning algorithms to obtain provably optimal policies for general risk-aware measures.

% Most of the literature looks at risk-averse cost-based methods. \citet{gangulyrisk} look at the utility-based approach and the risk-seeking setting in particular. They use the entropic value at risk measure and the paper focuses on how to accurately estimate this risk and optimize for it along with a parameterized policy. They use a multi-timescale-based approach to simultaneously learn the entropic value at risk and the parameterized policy. They prove asymptotic as well as finite time convergence results and provide an exhaustive set of experiments.

\subsection{A$^\ast$ variants}

There are several variants of the A$^\ast$ algorithm which involves using a heuristic function to optimistically estimate the remaining utility from any given part of the exponential search tree. This idea was first combined with Dec-POMDPs in Multi-agent A$^{\ast}$ \cite{szer2012maa}. They proposed trying several heuristic such as the MDP heuristic and centralized tighter but harder to compute POMDP heuristic and also a recursive heuristic using MAA$^{\ast}$. They present empirical results on small Dec-POMDPs. This idea is generalized by GMAA$^{\ast}$ \cite{oliehoek2008optimal} which breaks down the overall MAA$^{\ast}$ algorithm into fundamental steps denoted as ``Next'': which controls how partial policies are expanded and also returns their heuristic values; and ``Select'': which chooses the partial policy from the current policy pool to expand. They connect the MAA$^{\ast}$ problem to a tree search over collaborative Bayesian games. The GMAA$^{\ast}$-Incremental Clustering and Expansion \cite{oliehoek2013incremental} algorithm uses lossless incremental clustering to group histories together that satisfy the notion of probabilistic equivalence which is based on having the same joint distribution over the states and histories of the other agents. Once the clustering is done at a certain stage, the clusters themselves can be incrementally updated, rather than updating all the histories from scratch while retaining the lossless property (still an exact search). To further reduce the size of exponential tree search, the propose using incremental expansion which selectively expands the tree one node at a time based on the heuristic function.

\citet{koops2024approximate} introduced the Policy-Finding Multi-agent A$^{\ast}$ (PF-MAA) algorithm, which uses a sliding window of the past observations along with a novel heuristic that periodically uses the state which is tighter than the MDP heuristic. They cluster the sliding window observations \cite{oliehoek2009lossless, oliehoek2013incremental}. In addition, they propose the use of a heuristic for pruning nodes which are less promising in order to reduce the tree size of the exponential search. While an (approximately) optimal solution is not guaranteed by this method, they provide empirical results using large horizons for several Dec-POMDP problems which are very close to the optimal value.

\subsection{Exact and approximate exponential complexity algorithms}

\citet{dibangoye2016optimally} introduce a method called Feature-based Heuristic Search Value Iteration (FB-HSVI) and \citet{MahajanPhD} introduce a method called the designer's approach. They construct an MDP with the state space as the joint state, agent-state space and action space as the space of single step decision rules rather than individual joint actions of each agent. They propose using a dynamic programming formulation to obtain optimal policies. \citet{MahajanPhD} refers to this idea as the designer's approach (originally introduced by \citet{witsenhausen1973standard}), whereas \citet{dibangoye2016optimally} refers to it as the occupancy MDP. A distinction between the two approaches is that the dynamic program provided by \citet{MahajanPhD} gives the optimal solution within the given class of agent-state based policies, whereas \citet{dibangoye2016optimally} gives the history-based optimal solution. These methods have exponential time complexity. However, \citet{dibangoye2016optimally} groups similar histories/decision-rules together and uses finite history windows rather than the full history to make the occupancy MDP more compact. This makes it more efficient in practice in Dec-POMDPs with a lot of compact structure, however, it still remains of exponential time complexity and this would give sub-optimal solutions when there is important information beyond the finite history window. Another notable contribution is the use of feature based heuristics to further reduce the exponential search and provide upper and lower bounds, which lead to an anytime algorithm which is very useful in practice. This allows them to scale the results on several standard Dec-POMDPs.

Another idea is that of the common information approach \cite{nayyar2013decentralized} which relies on common information structure between agents to simplify the size of the centralized MDP. This helps in reducing the overall computation when their is common information, however, this still has exponential runtime complexity. A further variation on the common information approach has been considered in \citet{tang2023novel}.

\subsection{Sequential iterated best response algorithms}

Joint Equilibrium-based Search for Policies (JESP) \cite{nair2003taming} finds locally optimal joint policies by iteratively optimizing one agent's policy while keeping others fixed using a multi-agent-belief system. This approach, particularly the dynamic programming version (DP-JESP), leverages the piece-wise linear and convexity (PWLC) properties of the value function to provide exponential speedups over exhaustive search methods.

% \citet{dibangoye2016optimally} shows how approximately optimal solutions for Dec-POMDPs can be obtained, however, these are of exponential time complexity, and so, it is infeasible to extend these solutions to large time horizons. Our approach does not guarantee convergence to the global optimum (Nash) solution, but it does guarantee convergence to a locally optimal (Nash) solution and has polynomial space and time complexity.

% \subsection{Need to read}

% \cite{muller2021geometry}

\section{Memory and runtime analysis} \label{app:memory-runtime-analysis}

The update function involved in Algorithm~\ref{algo:policy-improvement-risk-seeking} involves storing a Q-table that depends linearly on the horizon $T$, which implies that the memory required by our algorithm is linear in the horizon $T$ of the problem rather than exponential. Furthermore, the runtime of computing the centralized Q-function $Q_t^{\vec \pi_t}$ is polynomial in the state, observation, action spaces and agent-state spaces of the individual agents in the Dec-POMDP, which is a significant reduction from exponential.

We consider Dec-POMDPs which have two agents and we use the the same agent-state size across them. This means that the memory and runtime complexity ends up being polynomial in the size of the agent state that we choose commonly for both agents. We show the memory metrics in the first two rows of Figure~\ref{fig:memory-runtime-results} and the runtime metrics over the next two rows. For each setting we took the mean and standard deviation of the $5$ random seeds considered.

\begin{figure*}[t] 
\centering
\setlength{\tabcolsep}{6pt} % more breathing room
\renewcommand{\arraystretch}{1.0}

\begin{tabular}{ccc}
% ================= Row 1: memory vs T =================
\includegraphics[width=0.30\textwidth]{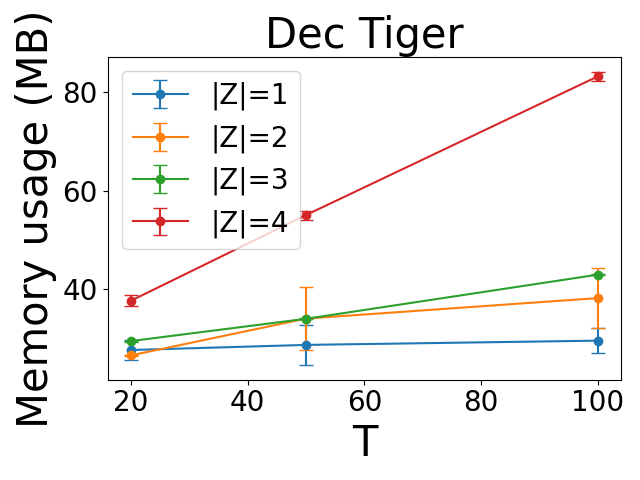} &
\includegraphics[width=0.30\textwidth]{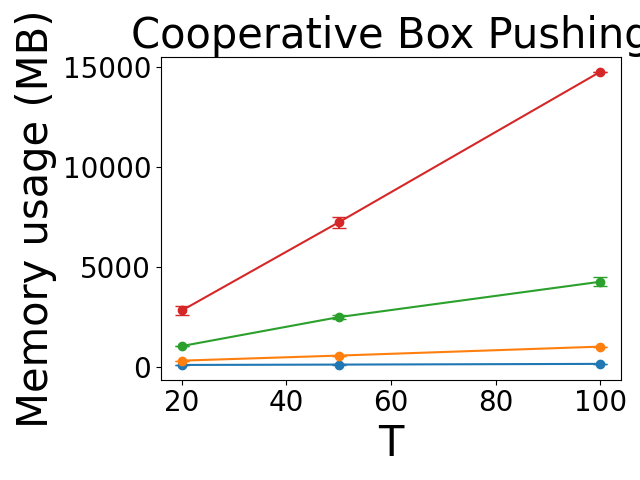} &
\includegraphics[width=0.30\textwidth]{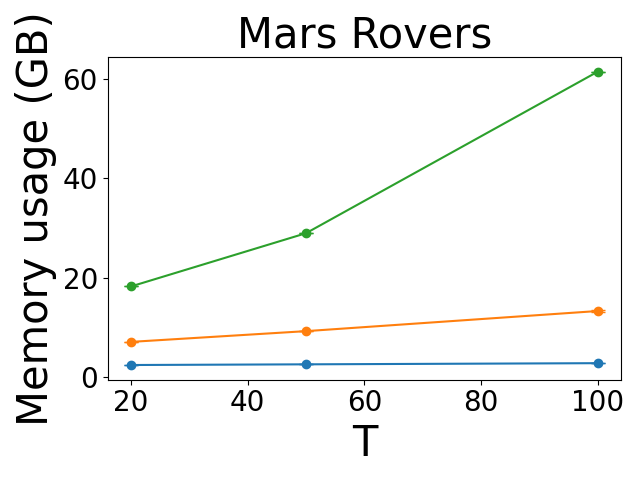} \\

% ================= Row 2: memory vs Z =================
\includegraphics[width=0.30\textwidth]{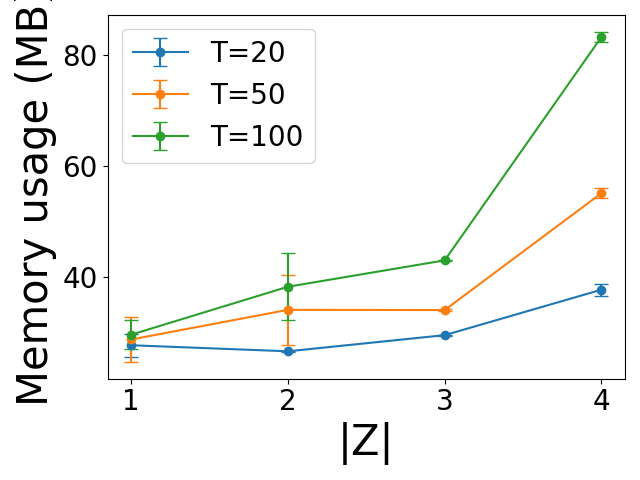} &
\includegraphics[width=0.30\textwidth]{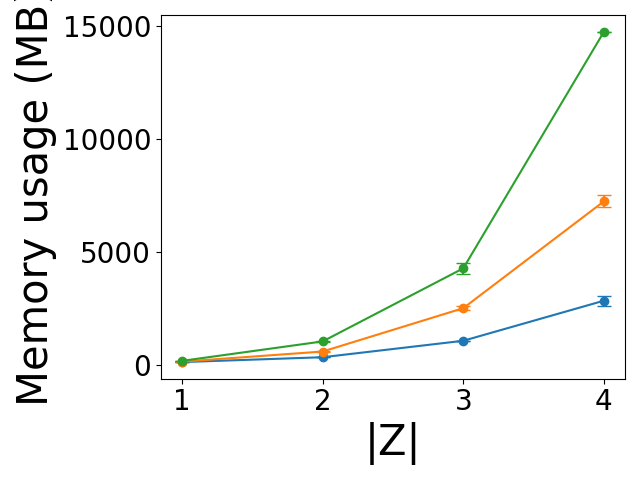} &
\includegraphics[width=0.30\textwidth]{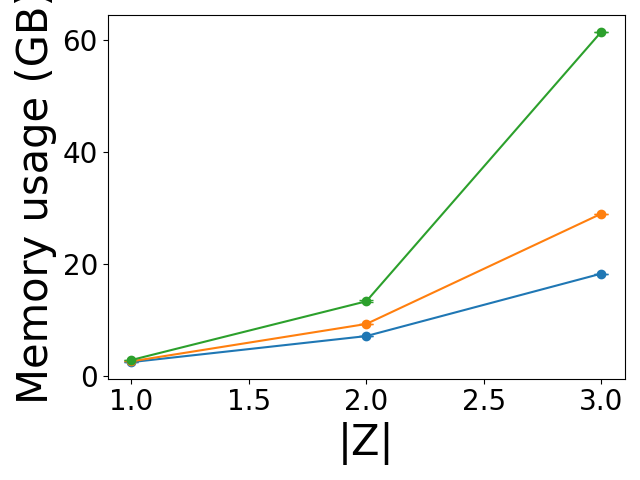} \\

% ================= Row 3: runtime vs T =================
\includegraphics[width=0.30\textwidth]{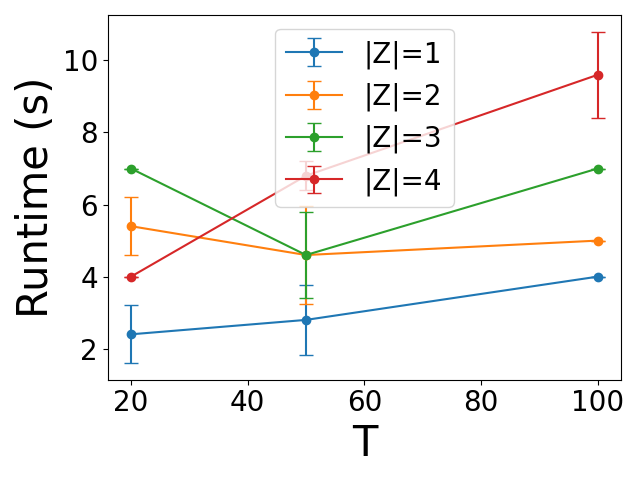} &
\includegraphics[width=0.30\textwidth]{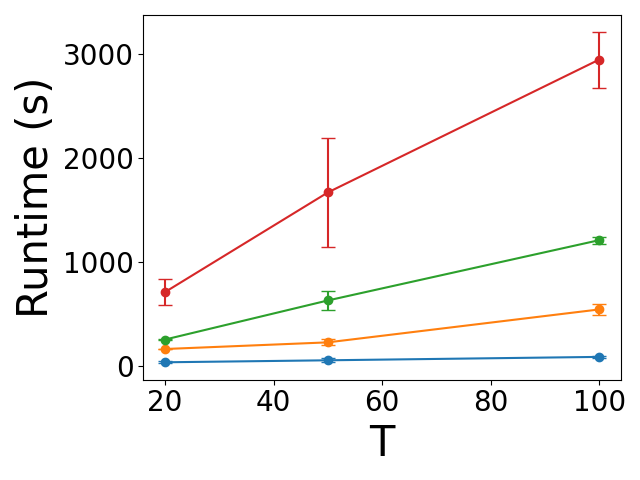} &
\includegraphics[width=0.30\textwidth]{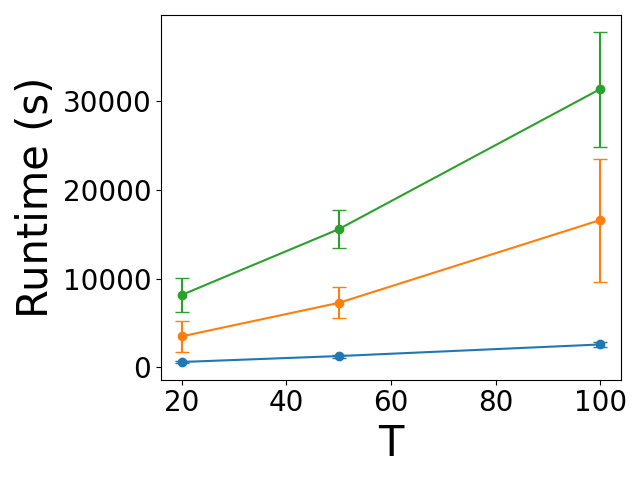} \\

% ================= Row 4: runtime vs Z =================
\includegraphics[width=0.30\textwidth]{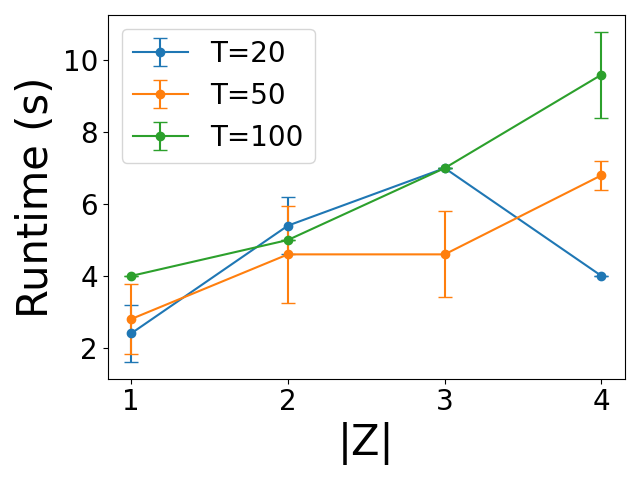} &
\includegraphics[width=0.30\textwidth]{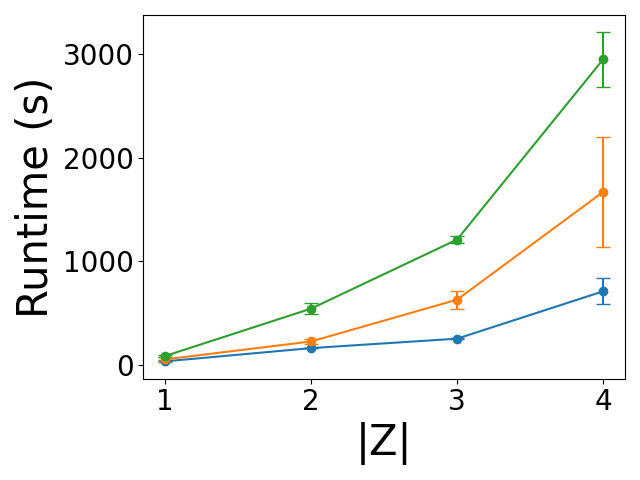} &
\includegraphics[width=0.30\textwidth]{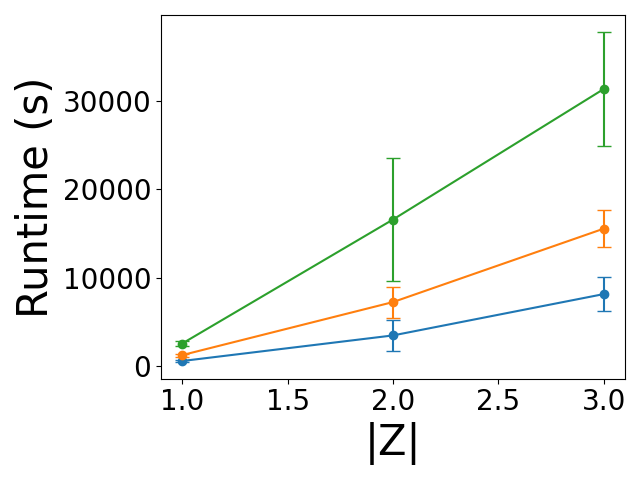}
\end{tabular}

\caption{
Scaling of memory and runtime with horizon $T$ and agent-state size $Z$ across three benchmarks:
Dec-Tiger, Cooperative Box Pushing, and Mars Rovers.
Rows correspond to different metrics, columns correspond to Dec-POMDPs.
}
\label{fig:memory-runtime-results}
\end{figure*}

\section{Interpreting Reactive Policies} \label{app:reactive-policy-descriptions}

\subsection{Cooperative Box Pushing Reactive Policy analysis}

\newcommand\emptyField{Empty}
\newcommand\wall{Wall}
\newcommand\otherAgent{OtherA}
\newcommand\smallBox{SmallB}
\newcommand\largeBox{LargeB}

\newcommand\moveForward{Forward}
\newcommand\turnLeft{TurnL}
\newcommand\turnRight{TurnR}
\newcommand\stay{Stay}

% \begin{figure}
%   \centering
%   \includegraphics[width=0.35\textwidth]{figures/boxpushing-illustration.png}
%   \caption{Cooperative Box Pushing (Image take from \citet{seuken2012improved})}
% \end{figure}

\begin{figure}[H] 
    \centering
    \includegraphics[width=0.3\textwidth]{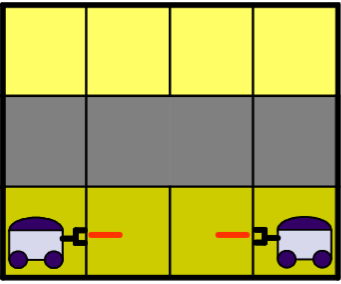}
    \caption{Cooperative Box Pushing (Image take from \citet{seuken2012improved}}
    \label{fig:boxpushing}
\end{figure}

In this Dec-POMDP, there are two agents, one in the left corner and the other in the right corner. They are initially facing the center of the arena (indicated by the red lines in Fig.~\ref{fig:boxpushing}). There is a small box denoted as \smallBox\ immediately adjacent to each of their starting position.

A larger box denoted as \largeBox\ lies in the middle. To reach this box the agents would have to move to the center. Agents can push a box by moving into it. The \smallBox\ can be moved by a single agent, however, the \largeBox\ requires both agents to simultaneously move into it. There is a larger reward for pushing the \largeBox\ and a minor reward for pushing the \smallBox\ , while agents are penalized for bumping into other objects.

The can turn its orientation using \turnLeft\ or \turnRight\ and can move one step in front of it by moving \moveForward\. It can also choose to \stay\ in the same spot (which is useful sometimes for cooperating with the other agent to push the \largeBox\ ). There is also some probability that agents fail to \moveForward\ or \turnLeft\ or \turnRight\ which makes the environment more challenging.

We present our learned RS-CPI reactive policy which is near-optimal in Table~\ref{table:cooperative-boxpushing-reactive-policy}. In order to understand how to read this policy and understand what is happening properly, we provide the following steps where we illustrate what happens at each time step. We trace out the trajectory of what happens when the movement actions of both agents always succeed. We mark this trajectory in bold in the Table~\ref{table:cooperative-boxpushing-reactive-policy}. We abbreviate Agent 1 as A1 (left) and Agent 2 as A2 (right).

\begin{enumerate}
    \item A1 sees \textbf{\emptyField} and chooses \textbf{\moveForward}. \\
    A2 sees \textbf{\emptyField} and chooses \textbf{\moveForward}.
    \item A1 sees \textbf{\otherAgent} and chooses \textbf{\turnLeft}.\\
    A2 sees \textbf{\otherAgent} and chooses \textbf{\turnRight}.\\
    (Both agents are now facing the \textbf{\largeBox})
    \item A1 sees \textbf{\largeBox} and chooses \textbf{\moveForward}.\\
    A2 sees \textbf{\largeBox} and chooses \textbf{\moveForward}.\\
    (Both agents move forward and get large reward for pushing the \textbf{\largeBox})
    \item The environment then gets reset back to the same initial state in this step, the actions do not matter in a reset step.
    \item It is interesting to note that same cycle from steps 1-3 continue in a periodic fashion.
\end{enumerate}

It might also be instructive to observe what happens in case an action fails. For example let's assume that when both agents do \moveForward\ , then A1 fails. We mark the actions in \blue{blue}. Then:
\begin{enumerate}
    \item A1 sees \textbf{\emptyField} and chooses \blue{\moveForward} (but fails).\\
    A2 sees \textbf{\emptyField} and chooses \blue{\moveForward} (success).
    \item A1 sees \textbf{\emptyField} and chooses \blue{\turnLeft}.\\
    A2 sees \textbf{\emptyField} and chooses \blue{\moveForward}.
    \item A1 sees \textbf{\smallBox} and chooses \blue{\moveForward}.\\
    A2 sees \textbf{\otherAgent} and chooses \blue{\turnRight}. \\
    (A small reward is received for pushing the small box successfully)
    \item This step is the rest step again.
    \item And then we get back into the cycle!
\end{enumerate}

So even when the best outcome doesn't work out, there is some contingency for the agents to finish the episode and restart it to try again. In this manner many different contingencies can be worked out, where the resulting actions are just the outcome of a process that overall has a tendency to maximize reward.

\begin{table*}
\centering
\small
\renewcommand{\arraystretch}{1.4}
\begin{tabular}{c|c|c|c|c|c}
\hline
Time &
\emptyField &
\wall &
\otherAgent &
\smallBox &
\largeBox \\
\hline

1 &
\begin{tabular}{c}\textbf{\moveForward} \\ \textbf{\moveForward}\end{tabular} &
\begin{tabular}{c}\turnLeft \\ \turnLeft\end{tabular} &
\begin{tabular}{c}\turnLeft \\ \turnLeft\end{tabular} &
\begin{tabular}{c}\turnLeft \\ \turnLeft\end{tabular} &
\begin{tabular}{c}\turnLeft \\ \turnLeft\end{tabular} \\

2 &
\begin{tabular}{c}\blue{\turnLeft} \\ \blue{\moveForward}\end{tabular} &
\begin{tabular}{c}\turnLeft \\ \turnRight\end{tabular} &
\begin{tabular}{c}\textbf{\turnLeft} \\ \textbf{\turnRight}\end{tabular} &
\begin{tabular}{c}\stay \\ \turnLeft\end{tabular} &
\begin{tabular}{c}\turnLeft \\ \turnLeft\end{tabular} \\

3 &
\begin{tabular}{c}\turnLeft \\ \moveForward\end{tabular} &
\begin{tabular}{c}\turnLeft \\ \turnRight\end{tabular} &
\begin{tabular}{c}\turnLeft \\ \blue{\turnRight}\end{tabular} &
\begin{tabular}{c}\blue{\moveForward} \\ \turnLeft\end{tabular} &
\begin{tabular}{c}\textbf{\moveForward} \\ \textbf{\moveForward}\end{tabular} \\

4 &
\begin{tabular}{c}\turnLeft \\ \moveForward\end{tabular} &
\begin{tabular}{c}\turnLeft \\ \turnRight\end{tabular} &
\begin{tabular}{c}\turnLeft \\ \turnRight\end{tabular} &
\begin{tabular}{c}\moveForward \\ \turnLeft\end{tabular} &
\begin{tabular}{c}\moveForward \\ \moveForward\end{tabular} \\

5 &
\begin{tabular}{c}\textbf{\moveForward} \\ \textbf{\moveForward}\end{tabular} &
\begin{tabular}{c}\turnLeft \\ \turnRight\end{tabular} &
\begin{tabular}{c}\turnLeft \\ \turnRight\end{tabular} &
\begin{tabular}{c}\moveForward \\ \moveForward\end{tabular} &
\begin{tabular}{c}\moveForward \\ \moveForward\end{tabular} \\

6 &
\begin{tabular}{c}\moveForward \\ \moveForward\end{tabular} &
\begin{tabular}{c}\turnLeft \\ \turnRight\end{tabular} &
\begin{tabular}{c}\textbf{\turnLeft} \\ \textbf{\turnRight}\end{tabular} &
\begin{tabular}{c}\stay \\ \stay\end{tabular} &
\begin{tabular}{c}\stay \\ \stay\end{tabular} \\

7 &
\begin{tabular}{c}\moveForward \\ \stay\end{tabular} &
\begin{tabular}{c}\stay \\ \turnRight\end{tabular} &
\begin{tabular}{c}\turnLeft \\ \turnRight\end{tabular} &
\begin{tabular}{c}\moveForward \\ \moveForward\end{tabular} &
\begin{tabular}{c}\textbf{\moveForward} \\ \textbf{\moveForward}\end{tabular} \\

8 &
\begin{tabular}{c}\turnLeft \\ \moveForward\end{tabular} &
\begin{tabular}{c}\turnLeft \\ \turnRight\end{tabular} &
\begin{tabular}{c}\turnLeft \\ \turnRight\end{tabular} &
\begin{tabular}{c}\moveForward \\ \moveForward\end{tabular} &
\begin{tabular}{c}\moveForward \\ \moveForward\end{tabular} \\

9 &
\begin{tabular}{c}\textbf{\moveForward} \\ \textbf{\moveForward}\end{tabular} &
\begin{tabular}{c}\turnLeft \\ \turnRight\end{tabular} &
\begin{tabular}{c}\turnLeft \\ \turnRight\end{tabular} &
\begin{tabular}{c}\stay \\ \stay\end{tabular} &
\begin{tabular}{c}\stay \\ \stay\end{tabular} \\

10 &
\begin{tabular}{c}\moveForward \\ \moveForward\end{tabular} &
\begin{tabular}{c}\turnLeft \\ \turnRight\end{tabular} &
\begin{tabular}{c}\textbf{\turnLeft} \\ \textbf{\turnRight}\end{tabular} &
\begin{tabular}{c}\moveForward \\ \moveForward\end{tabular} &
\begin{tabular}{c}\moveForward \\ \moveForward\end{tabular} \\

11 &
\begin{tabular}{c}\turnLeft \\ \moveForward\end{tabular} &
\begin{tabular}{c}\turnLeft \\ \turnRight\end{tabular} &
\begin{tabular}{c}\turnLeft \\ \turnRight\end{tabular} &
\begin{tabular}{c}\moveForward \\ \moveForward\end{tabular} &
\begin{tabular}{c}\textbf{\moveForward} \\ \textbf{\moveForward}\end{tabular} \\

12 &
\begin{tabular}{c}\moveForward \\ \moveForward\end{tabular} &
\begin{tabular}{c}\turnLeft \\ \turnRight\end{tabular} &
\begin{tabular}{c}\turnLeft \\ \turnRight\end{tabular} &
\begin{tabular}{c}\moveForward \\ \moveForward\end{tabular} &
\begin{tabular}{c}\moveForward \\ \moveForward\end{tabular} \\

13 &
\begin{tabular}{c}\textbf{\moveForward} \\ \textbf{\moveForward}\end{tabular} &
\begin{tabular}{c}\turnLeft \\ \turnRight\end{tabular} &
\begin{tabular}{c}\turnLeft \\ \turnRight\end{tabular} &
\begin{tabular}{c}\stay \\ \stay\end{tabular} &
\begin{tabular}{c}\stay \\ \stay\end{tabular} \\

14 &
\begin{tabular}{c}\moveForward \\ \moveForward\end{tabular} &
\begin{tabular}{c}\turnLeft \\ \turnRight\end{tabular} &
\begin{tabular}{c}\textbf{\turnLeft} \\ \textbf{\turnRight}\end{tabular} &
\begin{tabular}{c}\moveForward \\ \moveForward\end{tabular} &
\begin{tabular}{c}\moveForward \\ \moveForward\end{tabular} \\

15 &
\begin{tabular}{c}\turnLeft \\ \moveForward\end{tabular} &
\begin{tabular}{c}\turnLeft \\ \turnRight\end{tabular} &
\begin{tabular}{c}\turnLeft \\ \turnRight\end{tabular} &
\begin{tabular}{c}\moveForward \\ \moveForward\end{tabular} &
\begin{tabular}{c}\textbf{\moveForward} \\ \textbf{\moveForward}\end{tabular} \\

16 &
\begin{tabular}{c}\moveForward \\ \moveForward\end{tabular} &
\begin{tabular}{c}\turnLeft \\ \turnRight\end{tabular} &
\begin{tabular}{c}\turnLeft \\ \turnRight\end{tabular} &
\begin{tabular}{c}\moveForward \\ \moveForward\end{tabular} &
\begin{tabular}{c}\moveForward \\ \moveForward\end{tabular} \\

17 &
\begin{tabular}{c}\textbf{\moveForward} \\ \textbf{\moveForward}\end{tabular} &
\begin{tabular}{c}\turnLeft \\ \turnRight\end{tabular} &
\begin{tabular}{c}\turnLeft \\ \turnRight\end{tabular} &
\begin{tabular}{c}\moveForward \\ \moveForward\end{tabular} &
\begin{tabular}{c}\moveForward \\ \moveForward\end{tabular} \\

18 &
\begin{tabular}{c}\moveForward \\ \moveForward\end{tabular} &
\begin{tabular}{c}\turnLeft \\ \turnLeft\end{tabular} &
\begin{tabular}{c}\textbf{\turnLeft} \\ \textbf{\turnRight}\end{tabular} &
\begin{tabular}{c}\moveForward \\ \moveForward\end{tabular} &
\begin{tabular}{c}\moveForward \\ \moveForward\end{tabular} \\

19 &
\begin{tabular}{c}\turnLeft \\ \turnRight\end{tabular} &
\begin{tabular}{c}\turnRight \\ \turnLeft\end{tabular} &
\begin{tabular}{c}\turnLeft \\ \turnRight\end{tabular} &
\begin{tabular}{c}\moveForward \\ \moveForward\end{tabular} &
\begin{tabular}{c}\textbf{\moveForward} \\ \textbf{\moveForward}\end{tabular} \\

20 &
\begin{tabular}{c}\turnLeft \\ \turnLeft\end{tabular} &
\begin{tabular}{c}\turnLeft \\ \turnLeft\end{tabular} &
\begin{tabular}{c}\turnLeft \\ \turnLeft\end{tabular} &
\begin{tabular}{c}\moveForward \\ \moveForward\end{tabular} &
\begin{tabular}{c}\moveForward \\ \moveForward\end{tabular} \\

\hline

\end{tabular}
\caption{\textbf{Cooperative Box Pushing}: Joint policies over time. Top action in each cell: Agent~$1$, bottom action: Agent~$2$.}
\label{table:cooperative-boxpushing-reactive-policy}
\end{table*}

%Mars Rovers reactive policy
\subsection{Cooperative Box Pushing Reactive Policy analysis}

\newcommand\soneNS{Pos1-NS}
\newcommand\stwoNS{Pos2-NS}
\newcommand\sthreeNS{Pos3-NS}
\newcommand\sfourNS{Pos4-NS}
\newcommand\soneS{Pos1-S}
\newcommand\stwoS{Pos2-S}
\newcommand\sthreeS{Pos3-S}
\newcommand\sfourS{Pos-S}

\newcommand\RSsample{SAMPLE}
\newcommand\RSup{UP}
\newcommand\RSleft{LEFT}
\newcommand\RSright{RIGHT}
\newcommand\RSdown{DOWN}
\newcommand\RSdrill{DRILL}

This domain consists of two agents in a $2\times2$ grid. The agents can move \RSup\ , \RSdown\ , \RSleft\ , \RSright\ , take a \RSsample\ or \RSdrill\ at any of the $4$ locations. Two of the sites require only one agent to \RSsample\ them while two of the sites require both agents to coordinate and \RSdrill\ at the same time in order to get the maximum reward.

At any time the agents can observe their current position in addition to whether an experiment ( \RSsample\ or \RSdrill\ ) had previously been carried out at that particular location. The NS-suffix in the observation denotes that an experiment has not been done on the site and S-suffix denotes that an experiment has been done. If a site needs to be \RSsample\ , but \RSdrill\ is done instead,  then it incurs a large penalty for ruining the sample, whereas doing the opposite gives a small positive reward. Performing an experiment at any location more than once incurs negative rewards.

We present our learned RS-CPI reactive policy which is near-optimal in Table~\ref{table:mars-rovers-reactive-policy}. A similar approach can be followed to the Cooperative Box Pushing example to analyze the policy given in Table~\ref{table:mars-rovers-reactive-policy}.

% \begin{figure}
%   \centering
%   \includegraphics[width=0.35\textwidth]{figures/boxpushing-illustration.png}
%   \caption{Cooperative Box Pushing (Image take from \citet{seuken2012improved})}
% \end{figure}

% \begin{figure}[H] 
%     \centering
%     \includegraphics[width=0.3\textwidth]{figures/boxpushing-illustration.png}
%     \caption{Cooperative Box Pushing (Image take from \citet{seuken2012improved}}
%     \label{fig:boxpushing}
% \end{figure}

\begin{table*}
\centering
\scriptsize
\setlength{\tabcolsep}{3pt}
\renewcommand{\arraystretch}{1.25}

\begin{tabular}{c|cccccccc}
\hline
$t$ &
\soneNS & \stwoNS & \sthreeNS & \sfourNS &
\soneS & \stwoS & \sthreeS & \sfourS \\
\hline

1  &
\begin{tabular}{c}\RSsample\\\RSright\end{tabular} &
\begin{tabular}{c}\RSup\\\RSup\end{tabular} &
\begin{tabular}{c}\RSup\\\RSup\end{tabular} &
\begin{tabular}{c}\RSup\\\RSup\end{tabular} &
\begin{tabular}{c}\RSup\\\RSup\end{tabular} &
\begin{tabular}{c}\RSup\\\RSup\end{tabular} &
\begin{tabular}{c}\RSup\\\RSup\end{tabular} &
\begin{tabular}{c}\RSup\\\RSup\end{tabular} \\

2  &
\begin{tabular}{c}\RSsample\\\RSsample\end{tabular} &
\begin{tabular}{c}\RSsample\\\RSup\end{tabular} &
\begin{tabular}{c}\RSleft\\\RSsample\end{tabular} &
\begin{tabular}{c}\RSup\\\RSup\end{tabular} &
\begin{tabular}{c}\RSdown\\\RSright\end{tabular} &
\begin{tabular}{c}\RSup\\\RSup\end{tabular} &
\begin{tabular}{c}\RSup\\\RSup\end{tabular} &
\begin{tabular}{c}\RSup\\\RSup\end{tabular} \\

3  &
\begin{tabular}{c}\RSsample\\\RSright\end{tabular} &
\begin{tabular}{c}\RSsample\\\RSsample\end{tabular} &
\begin{tabular}{c}\RSleft\\\RSsample\end{tabular} &
\begin{tabular}{c}\RSleft\\\RSup\end{tabular} &
\begin{tabular}{c}\RSdown\\\RSright\end{tabular} &
\begin{tabular}{c}\RSup\\\RSdown\end{tabular} &
\begin{tabular}{c}\RSleft\\\RSdown\end{tabular} &
\begin{tabular}{c}\RSdown\\\RSup\end{tabular} \\

4  &
\begin{tabular}{c}\RSsample\\\RSsample\end{tabular} &
\begin{tabular}{c}\RSsample\\\RSright\end{tabular} &
\begin{tabular}{c}\RSleft\\\RSsample\end{tabular} &
\begin{tabular}{c}\RSleft\\\RSsample\end{tabular} &
\begin{tabular}{c}\RSdown\\\RSright\end{tabular} &
\begin{tabular}{c}\RSright\\\RSright\end{tabular} &
\begin{tabular}{c}\RSleft\\\RSdown\end{tabular} &
\begin{tabular}{c}\RSdown\\\RSup\end{tabular} \\

5  &
\begin{tabular}{c}\RSdrill\\\RSdrill\end{tabular} &
\begin{tabular}{c}\RSdown\\\RSup\end{tabular} &
\begin{tabular}{c}\RSup\\\RSdown\end{tabular} &
\begin{tabular}{c}\RSdown\\\RSdown\end{tabular} &
\begin{tabular}{c}\RSdown\\\RSright\end{tabular} &
\begin{tabular}{c}\RSright\\\RSright\end{tabular} &
\begin{tabular}{c}\RSleft\\\RSdown\end{tabular} &
\begin{tabular}{c}\RSup\\\RSdown\end{tabular} \\

6  &
\begin{tabular}{c}\RSdown\\\RSsample\end{tabular} &
\begin{tabular}{c}\RSsample\\\RSright\end{tabular} &
\begin{tabular}{c}\RSsample\\\RSsample\end{tabular} &
\begin{tabular}{c}\RSdrill\\\RSdrill\end{tabular} &
\begin{tabular}{c}\RSdown\\\RSright\end{tabular} &
\begin{tabular}{c}\RSsample\\\RSright\end{tabular} &
\begin{tabular}{c}\RSsample\\\RSdown\end{tabular} &
\begin{tabular}{c}\RSup\\\RSleft\end{tabular} \\

7  &
\begin{tabular}{c}\RSup\\\RSup\end{tabular} &
\begin{tabular}{c}\RSsample\\\RSsample\end{tabular} &
\begin{tabular}{c}\RSleft\\\RSsample\end{tabular} &
\begin{tabular}{c}\RSdown\\\RSsample\end{tabular} &
\begin{tabular}{c}\RSdown\\\RSright\end{tabular} &
\begin{tabular}{c}\RSright\\\RSup\end{tabular} &
\begin{tabular}{c}\RSleft\\\RSdown\end{tabular} &
\begin{tabular}{c}\RSdown\\\RSup\end{tabular} \\

8  &
\begin{tabular}{c}\RSdrill\\\RSdrill\end{tabular} &
\begin{tabular}{c}\RSdown\\\RSdown\end{tabular} &
\begin{tabular}{c}\RSsample\\\RSup\end{tabular} &
\begin{tabular}{c}\RSdown\\\RSdown\end{tabular} &
\begin{tabular}{c}\RSdown\\\RSright\end{tabular} &
\begin{tabular}{c}\RSright\\\RSright\end{tabular} &
\begin{tabular}{c}\RSleft\\\RSdown\end{tabular} &
\begin{tabular}{c}\RSup\\\RSup\end{tabular} \\

9  &
\begin{tabular}{c}\RSdrill\\\RSdrill\end{tabular} &
\begin{tabular}{c}\RSsample\\\RSsample\end{tabular} &
\begin{tabular}{c}\RSsample\\\RSsample\end{tabular} &
\begin{tabular}{c}\RSdrill\\\RSdrill\end{tabular} &
\begin{tabular}{c}\RSdown\\\RSright\end{tabular} &
\begin{tabular}{c}\RSsample\\\RSright\end{tabular} &
\begin{tabular}{c}\RSleft\\\RSdown\end{tabular} &
\begin{tabular}{c}\RSdown\\\RSup\end{tabular} \\

10 &
\begin{tabular}{c}\RSdrill\\\RSdrill\end{tabular} &
\begin{tabular}{c}\RSsample\\\RSright\end{tabular} &
\begin{tabular}{c}\RSleft\\\RSsample\end{tabular} &
\begin{tabular}{c}\RSdown\\\RSdown\end{tabular} &
\begin{tabular}{c}\RSdown\\\RSright\end{tabular} &
\begin{tabular}{c}\RSright\\\RSright\end{tabular} &
\begin{tabular}{c}\RSleft\\\RSdown\end{tabular} &
\begin{tabular}{c}\RSdown\\\RSdown\end{tabular} \\

11 &
\begin{tabular}{c}\RSsample\\\RSright\end{tabular} &
\begin{tabular}{c}\RSdown\\\RSright\end{tabular} &
\begin{tabular}{c}\RSleft\\\RSdown\end{tabular} &
\begin{tabular}{c}\RSdown\\\RSdown\end{tabular} &
\begin{tabular}{c}\RSdown\\\RSright\end{tabular} &
\begin{tabular}{c}\RSright\\\RSright\end{tabular} &
\begin{tabular}{c}\RSdown\\\RSdown\end{tabular} &
\begin{tabular}{c}\RSdown\\\RSup\end{tabular} \\

12 &
\begin{tabular}{c}\RSsample\\\RSright\end{tabular} &
\begin{tabular}{c}\RSsample\\\RSup\end{tabular} &
\begin{tabular}{c}\RSdown\\\RSsample\end{tabular} &
\begin{tabular}{c}\RSdrill\\\RSdrill\end{tabular} &
\begin{tabular}{c}\RSdown\\\RSright\end{tabular} &
\begin{tabular}{c}\RSsample\\\RSup\end{tabular} &
\begin{tabular}{c}\RSsample\\\RSdown\end{tabular} &
\begin{tabular}{c}\RSup\\\RSup\end{tabular} \\

13 &
\begin{tabular}{c}\RSdrill\\\RSdrill\end{tabular} &
\begin{tabular}{c}\RSdown\\\RSright\end{tabular} &
\begin{tabular}{c}\RSup\\\RSup\end{tabular} &
\begin{tabular}{c}\RSleft\\\RSdown\end{tabular} &
\begin{tabular}{c}\RSdown\\\RSright\end{tabular} &
\begin{tabular}{c}\RSright\\\RSright\end{tabular} &
\begin{tabular}{c}\RSdown\\\RSdown\end{tabular} &
\begin{tabular}{c}\RSup\\\RSup\end{tabular} \\

14 &
\begin{tabular}{c}\RSdrill\\\RSdrill\end{tabular} &
\begin{tabular}{c}\RSsample\\\RSsample\end{tabular} &
\begin{tabular}{c}\RSsample\\\RSsample\end{tabular} &
\begin{tabular}{c}\RSdrill\\\RSdrill\end{tabular} &
\begin{tabular}{c}\RSright\\\RSdown\end{tabular} &
\begin{tabular}{c}\RSsample\\\RSright\end{tabular} &
\begin{tabular}{c}\RSsample\\\RSdown\end{tabular} &
\begin{tabular}{c}\RSup\\\RSup\end{tabular} \\

15 &
\begin{tabular}{c}\RSdrill\\\RSdrill\end{tabular} &
\begin{tabular}{c}\RSup\\\RSsample\end{tabular} &
\begin{tabular}{c}\RSsample\\\RSsample\end{tabular} &
\begin{tabular}{c}\RSdown\\\RSdown\end{tabular} &
\begin{tabular}{c}\RSright\\\RSdown\end{tabular} &
\begin{tabular}{c}\RSright\\\RSright\end{tabular} &
\begin{tabular}{c}\RSdown\\\RSdown\end{tabular} &
\begin{tabular}{c}\RSup\\\RSdown\end{tabular} \\

16 &
\begin{tabular}{c}\RSsample\\\RSdown\end{tabular} &
\begin{tabular}{c}\RSup\\\RSright\end{tabular} &
\begin{tabular}{c}\RSup\\\RSdown\end{tabular} &
\begin{tabular}{c}\RSdown\\\RSdown\end{tabular} &
\begin{tabular}{c}\RSright\\\RSdown\end{tabular} &
\begin{tabular}{c}\RSright\\\RSright\end{tabular} &
\begin{tabular}{c}\RSdown\\\RSdown\end{tabular} &
\begin{tabular}{c}\RSdown\\\RSleft\end{tabular} \\

17 &
\begin{tabular}{c}\RSup\\\RSup\end{tabular} &
\begin{tabular}{c}\RSup\\\RSsample\end{tabular} &
\begin{tabular}{c}\RSsample\\\RSdown\end{tabular} &
\begin{tabular}{c}\RSdrill\\\RSdrill\end{tabular} &
\begin{tabular}{c}\RSright\\\RSdown\end{tabular} &
\begin{tabular}{c}\RSsample\\\RSright\end{tabular} &
\begin{tabular}{c}\RSsample\\\RSdown\end{tabular} &
\begin{tabular}{c}\RSup\\\RSleft\end{tabular} \\

18 &
\begin{tabular}{c}\RSdrill\\\RSdrill\end{tabular} &
\begin{tabular}{c}\RSsample\\\RSdown\end{tabular} &
\begin{tabular}{c}\RSup\\\RSsample\end{tabular} &
\begin{tabular}{c}\RSdown\\\RSdown\end{tabular} &
\begin{tabular}{c}\RSright\\\RSdown\end{tabular} &
\begin{tabular}{c}\RSright\\\RSright\end{tabular} &
\begin{tabular}{c}\RSdown\\\RSdown\end{tabular} &
\begin{tabular}{c}\RSleft\\\RSleft\end{tabular} \\

19 &
\begin{tabular}{c}\RSup\\\RSup\end{tabular} &
\begin{tabular}{c}\RSsample\\\RSsample\end{tabular} &
\begin{tabular}{c}\RSsample\\\RSsample\end{tabular} &
\begin{tabular}{c}\RSdrill\\\RSdrill\end{tabular} &
\begin{tabular}{c}\RSdown\\\RSright\end{tabular} &
\begin{tabular}{c}\RSup\\\RSright\end{tabular} &
\begin{tabular}{c}\RSsample\\\RSdown\end{tabular} &
\begin{tabular}{c}\RSleft\\\RSleft\end{tabular} \\

20 &
\begin{tabular}{c}\RSdrill\\\RSdrill\end{tabular} &
\begin{tabular}{c}\RSsample\\\RSsample\end{tabular} &
\begin{tabular}{c}\RSsample\\\RSsample\end{tabular} &
\begin{tabular}{c}\RSsample\\\RSsample\end{tabular} &
\begin{tabular}{c}\RSup\\\RSup\end{tabular} &
\begin{tabular}{c}\RSup\\\RSup\end{tabular} &
\begin{tabular}{c}\RSup\\\RSup\end{tabular} &
\begin{tabular}{c}\RSup\\\RSup\end{tabular} \\

\hline
\end{tabular}

\caption{\textbf{Mars Rovers}: Joint policies over time. Top action in each cell: Agent~$1$, bottom action: Agent~$2$.}
\label{table:mars-rovers-reactive-policy}
\end{table*}

\section{Proof of Proposition~\ref{prop:improvement}} \label{app:main-proof}

\newcommand\SMALLVEC[2]{\left[ \begin{smallmatrix} 
    #1 \\ #2
  \end{smallmatrix} \right]}

\newcommand\OBS{\SMALLVEC{y}{z_{-}}}
\newcommand\OBSi{\SMALLVEC{y^i}{z^i_{-}}}
\newcommand\OBSj{\SMALLVEC{y^{-i}}{z^{-i}_{-}}}
\newcommand\NOBS{\SMALLVEC{y_{+}}{z}}

\newcommand\ACT{\SMALLVEC{a}{z}}
\newcommand\ACTi{\SMALLVEC{a^i}{z^i}}
\newcommand\ACTj{\SMALLVEC{a^{-i}}{z^{-i}}}
\newcommand\NACT{\SMALLVEC{a_{+}}{z_{+}}}

\textbf{Risk-seeking centralized value and action-value functions.} We start by defining centralized value and action-value functions for risk-seeking evaluation of a policy $\pi$ with risk parameter $\lambda$, $\lambda > 0$. 
\begin{align*}
    Q^{\vec \pi_{\lambda,T}}_{\lambda, T}\bigl(s, \OBS, \ACT\bigr) &= r(s,a)
    \\
    \intertext{and for $t \in \{T-1, \dots, 1\}$, we have}
    Q^{\vec \pi_t}_{\lambda,t}\bigl(s, \OBS, \ACT\bigr) &= r(s,a) 
    + 
    \tfrac{1}{\lambda}
    \log \biggl[
    \sum_{s_{+}, y_{+}} P(s_{+}, y_{+} \mid s, a) 
    \exp\bigl( 
    \lambda V^{\pi_{t+1}\vec \pi_{t+1}}_{\lambda, t+1}\bigl(s_{+}, \NOBS\bigr)
    \bigr)
    \biggr]
    \\
    \shortintertext{where}
    V^{\pi_t, \vec \pi_t}_{\lambda,t}\bigl(s, \OBS \bigr)
    &= \tfrac {1}{\lambda} \log \biggl[ \sum_{a,z} \pi_t(a,z \mid y, z_{-}) 
    \exp\bigl(\lambda Q^{\vec \pi_t}_{\lambda,t}\bigl( s, \OBS, \ACT \bigr) \bigr)
    \biggr]
    \\
    % \shortintertext{\red{Then}}
    % \red{Q^{\vec \pi_t}_{\lambda,t}\bigl(s, \OBS, \ACT\bigr)} &
    % \red{= r(s,a) 
    % + 
    % \tfrac{1}{\lambda}
    % \log \biggl[
    % \sum_{s_{+}, y_{+}, a_{+}, z_{+}} P(s_{+}, y_{+} \mid s, a) \pi_{t+1}(a_{+}, z_{+} \mid y_{+}, z)
    % \exp\bigl( 
    % \lambda Q^{\pi_{t+1}\vec \pi_{t+1}}_{\lambda, t+1}\bigl(s_{+}, \NOBS, \NACT\bigr)
    % \bigr)
    % \biggr]}
\end{align*}

Based on this risk-seeking policy evaluation, define a risk-seeking performance as
\[
    J_{1:T}(\boldsymbol{\pi}) =
    \REXP\biggl[ \sum_{t=1}^T r(S_t, A_t) \biggr] 
    =
    \sum_{s_1, y_1, z_0} \zeta_1(s_1, y_1) \phi(z_0) V^{\boldsymbol{\pi}}_{\lambda,1}\bigl(s_1, \SMALLVEC{y_1}{z_0} \bigr)
\]
where $\phi(z_0) = \phi^1(z^1_0) \phi^2(z^2_0)$ (recall that $\phi^i$ is the distribution of the initial agent state of agent~$i$).

% We also define the expected risk-seeking performance from time~$t$ onward as follows:
% \[
%    J_{\lambda}_{t:T}(\cev \pi_t, \pi_t, \vec \pi_t) =
%    \sum_{s_t, y_t, z_{t-1}, a_t, z_t} \zeta^{\cev \pi_t}_t(s_t, y_t, z_t
% \]

Recall that the local average risk-seeking action-value functions is defined as
\begin{align*}
    \bar Q^{(\cev \pi_t, \pi^{-i}_t, \vec \pi_t)}_{\lambda,t}\Bigl(\OBSi, \ACTi\Bigr) &=
    \tfrac {1}{\lambda} \log \biggl[ 
    \sum_{s, y^{-i}, z^{-i}_{-}, a^{-i}, z^{-i}}
    \zeta^{\cev \pi_t, \pi^{-i}_t}_t\Bigl( s, \OBSj, \ACTj \Bigm| \OBSi\Bigr)
    \exp\bigl( \lambda Q^{\vec \pi_t}_{\lambda,t}\bigl(s, \OBS, \ACT\bigr)\bigr)
    \biggr]
\end{align*}

%%% USE THE SAME MACROS!!
\renewcommand\OBSi{\SMALLVEC{y^1}{z^1_{-}}}
\renewcommand\OBSj{\SMALLVEC{y^{2}}{z^{2}_{-}}}

\renewcommand\ACTi{\SMALLVEC{a^1}{z^1}}
\renewcommand\ACTj{\SMALLVEC{a^{2}}{z^{2}}}

    By definition of $\bar \pi^2_t$ (which is deterministic), for all $(y^2, z^2_{-}) \in \ALPHABET Y \times \ALPHABET Z$, 
    \[
       \bar Q^{(\cev \pi_t, \pi^1_t, \vec \pi_t)}_{\lambda, t} \Bigl( \OBSj, \bar \pi^2_t\Bigl( \OBSj \Bigr) \Bigr) 
       \ge 
       \bar Q^{(\cev \pi_t, \pi^1_t, \vec \pi_t)}_{\lambda, t} \Bigl( \OBSj, \ACTj \Bigr), 
       \quad \forall \ACTj \in \ALPHABET{A}^2 \times \ALPHABET{Z}^2
    \]
    Moreover, since $\lambda > 0$, 
    \[
     \exp\Bigl( \lambda  
       \bar Q^{(\cev \pi_t, \pi^1_t, \vec \pi_t)}_{\lambda, t} \Bigl( \OBSj, \bar \pi^2_t\Bigl( \OBSj \Bigr) \Bigr) 
    \Bigr)
       \ge 
     \exp\Bigl( \lambda  
       \bar Q^{(\cev \pi_t, \pi^1_t, \vec \pi_t)}_{\lambda, t} \Bigl( \OBSj, \ACTj \Bigr)
    \Bigr),
       \quad \forall \ACTj \in \ALPHABET{A}^2 \times \ALPHABET{Z}^2
    \]
    Hence,
    \[
     \exp\Bigl( \lambda  
       \bar Q^{(\cev \pi_t, \pi^1_t, \vec \pi_t)}_{\lambda, t} \Bigl( \OBSj, \bar \pi^2_t\Bigl( \OBSj \Bigr) \Bigr) 
    \Bigr)
       \ge
            \sum_{a^2,z^2} \pi^2_t \Bigl( \ACTj \Bigm| \OBSj \Bigr) 
     \exp\Bigl( \lambda  
       \bar Q^{(\cev \pi_t, \pi^1_t, \vec \pi_t)}_{\lambda, t} \Bigl( \OBSj, \ACTj \Bigr)
    \Bigr)
    \]

    Substituting the definition of the averaged action-value functions, we get
    \begin{align}
        & 
    \sum_{s, y^{1}, z^{1}_{-}, a^{1}, z^{1}}
    \zeta^{\cev \pi_t, \pi^{1}_t}_t\Bigl( s, \OBSi, \ACTi \Bigm| \OBSj\Bigr)
    \exp\Bigl( \lambda Q^{\vec \pi_t}_{\lambda,t}\Bigl(s, \OBSi, \OBSj, \ACTi, \bar \pi^2_t\Bigl(\OBSj\Bigr)\Bigr)\Bigr)
    \notag \\
    &\ge
     \sum_{a^2,z^2} \pi^2_t \Bigl( \ACTj \Bigm| \OBSj \Bigr) 
    \sum_{s, y^{1}, z^{1}_{-}, a^{1}, z^{1}}
    \zeta^{\cev \pi_t, \pi^{1}_t}_t\Bigl( s, \OBSi, \ACTi \Bigm| \OBSj\Bigr)
    \exp\Bigl( \lambda Q^{\vec \pi_t}_{\lambda,t}\Bigl(s, \OBSi, \OBSj, \ACTi, \ACTj\Bigr)\Bigr)
    \notag \\
    &=
    \sum_{s, y^{1}, z^{1}_{-}, a^{1}, z^{1}, a^2, z^2}
    \zeta^{\cev \pi_t, \pi^{1}_t, \pi^2_t}_t\Bigl( s, \OBSi, \ACTi, \ACTj \Bigm| \OBSj\Bigr)
    \exp\Bigl( \lambda Q^{\vec \pi_t}_{\lambda,t}\Bigl(s, \OBSi, \OBSj, \ACTi, \ACTj\Bigr)\Bigr)
    \label{eq:ineq-1}
    \end{align}
    Observe that
    \[
    \zeta^{\cev \pi_t, \pi^{1}_t, \pi^2_t}_t\Bigl( s, \OBSi, \ACTi, \ACTj \Bigm| \OBSj\Bigr)
    =
    \frac{\zeta^{\cev \pi_t, \pi^1_t, \pi^2_t}_t\Bigl( s, \OBSi, \OBSj, \ACTi, \ACTj\Bigr)}
         {\zeta^{\cev \pi_t}_t\Bigl( \OBSj \Bigr)}.
    \]
    Substituting this in~\eqref{eq:ineq-1} and canceling the denominator $\zeta^{\vec \pi_t}_t(y^2,z^2_{-})$ from both sides, we get
    \begin{align}
        & 
    \sum_{s, y^{1}, z^{1}_{-}, a^{1}, z^{1}, a^2, z^2}
    \zeta^{\cev \pi_t, \pi^{1}_t, \bar \pi^2_t}_t\Bigl( s, \OBSi, \OBSj, \ACTi, \ACTj \Bigr)
    \exp\Bigl( \lambda Q^{\vec \pi_t}_{\lambda,t}\Bigl(s, \OBSi, \OBSj, \ACTi, \ACTj\Bigr)\Bigr)
    \notag \\
    &\ge
    \sum_{s, y^{1}, z^{1}_{-}, a^{1}, z^{1}, a^2, z^2}
    \zeta^{\cev \pi_t, \pi^{1}_t, \pi^2_t}_t\Bigl( s, \OBSi, \OBSj, \ACTi, \ACTj \Bigr)
    \exp\Bigl( \lambda Q^{\vec \pi_t}_{\lambda,t}\Bigl(s, \OBSi, \OBSj, \ACTi, \ACTj\Bigr)\Bigr)
    \label{eq:ineq-2}
    \end{align}
    where in the LHS we are using the fact that $\bar \pi^2_t$ is a deterministic policy and therefore assigns positive weight to only $\bar \pi^2_t(y^2,z^2_{-})$.

    Recall that the expected risk-seeking performance from time~$t$ onward is given by
    \begin{align*}
        &J_{t:T}(\cev \pi_t, \pi^1_t, \pi^2_t, \vec \pi_t)
        =  \REXP\biggl[ \sum_{t=1}^T r(S_t, A_t) \biggr] 
        \\
    &= \tfrac 1{\lambda} \log \biggl[ 
    \sum_{s, y^{1}, z^{1}_{-}, a^{1}, z^{1}, a^2, z^2}
    \zeta^{\cev \pi_t, \pi^{1}_t, \pi^2_t}_t\Bigl( s, \OBSi, \OBSj, \ACTi, \ACTj \Bigr)
    \exp\Bigl( \lambda Q^{\vec \pi_t}_{\lambda,t}\Bigl(s, \OBSi, \OBSj, \ACTi, \ACTj\Bigr)\Bigr)
    \biggr]
    \end{align*}
    Then, since $\lambda > 0$, \eqref{eq:ineq-2} implies that 
    \[
        J_{t:T}(\cev \pi_t, \pi^1_t, \bar \pi^2_t, \vec \pi_t)
        \ge 
        J_{t:T}(\cev \pi_t, \pi^1_t, \pi^2_t, \vec \pi_t).
    \] 

    The result then follows from the observation that for any policy $\pi$
    \[
        J_{1:T}(\pi) = J_{1:t-1}(\cev \pi_t) + J_{t:T}(\cev \pi_t, \pi_t, \vec \pi_t).
    \]

\section{Proof of Corollary~\ref{cor:improvement}} \label{app:cor-proof}

\renewcommand\OBS{\SMALLVEC{y}{z_{-}}}
\renewcommand\OBSi{\SMALLVEC{y^i}{z^i_{-}}}
\renewcommand\OBSj{\SMALLVEC{y^{-i}}{z^{-i}_{-}}}
\renewcommand\NOBS{\SMALLVEC{y_{+}}{z}}

\renewcommand\ACT{\SMALLVEC{a}{z}}
\renewcommand\ACTi{\SMALLVEC{a^i}{z^i}}
\renewcommand\ACTj{\SMALLVEC{a^{-i}}{z^{-i}}}
\renewcommand\NACT{\SMALLVEC{a_{+}}{z_{+}}}

\begin{align*}
    \pi^{i, \alpha}_t \in \GREEDY^{i, \alpha}_t(\boldsymbol{\pi}) \coloneqq (1-\alpha) \pi^i_t + \alpha \GREEDY^i_t(\boldsymbol{\pi}).
\end{align*}

\begin{align*}
    J_{t:T}(\cev \pi_t,  \pi^{i, \alpha}_t, \pi^{-i}_t, \vec \pi_t) &= (1-\alpha) \sum_{s, y, z_{-}, a, z} 
    \zeta^{\cev \pi_t, \pi^{-i}_t, \bar \pi^i_t}_t\Bigl( s, \ACT \Bigm| \OBS \Bigr)
    \exp\bigl( \lambda Q^{\vec \pi_t}_{\lambda,t}\bigl(s, \OBS, \ACT \bigr)\bigr) \\
    &\quad + \alpha \sum_{s, y, z_{-}, a, z} \zeta^{\cev \pi_t, \pi^{-i}_t, \pi^i_t}_t\Bigl( s, \ACT \Bigm| \OBS \Bigr)
    \exp\bigl( \lambda Q^{\vec \pi_t}_{\lambda,t}\bigl(s, \OBS, \ACT \bigr)\bigr)
\end{align*}
Since we have that
\begin{align*}
J_{t:T}(\cev \pi_t, \bar \pi^i_t, \pi^{-i}_t, \vec \pi_t)
\ge
J_{t:T}(\cev \pi_t, \pi_t, \vec \pi_t),
\end{align*}
we have by linearity that 
\begin{align*}
J_{t:T}(\cev \pi_t,  \pi^{i, \alpha}_t, \pi^{-i}_t, \vec \pi_t)
\ge
J_{t:T}(\cev \pi_t, \pi_t, \vec \pi_t).
\end{align*}

\section{Proof of Proposition \ref{prop:convergence}} \label{app:prop-convg-proof}

Corollary~\ref{cor:improvement} implies that $\{ J_{1:T}(\boldsymbol{\pi}^{(k)}) \}_{k \ge 0}$ is a non-decreasing sequence which is uniformly bounded (since the per-step rewards are bounded). Therefore, the sequence $J_{1:T}(\boldsymbol{\pi}^{(k)})$ converges. When conservatism is absent (i.e., $\alpha^{(k)} = 1$ for all $k$), the policies $\{ \boldsymbol{\pi}^{(k)} \}_{k \ge 1}$ are deterministic (assume that ties in $\argmax$ are broken by an arbitrary but fixed rule). Since there are a finite number of deterministic policies, the convergence must happen in a finite number of steps.

We now show that the limit policy will satisfy~\eqref{eq:local-optimality}. Using Corollary~\ref{cor:improvement},  we can show that for any $t$, $J_{t:T} (\boldsymbol{\pi}^{(k)})$ also converges to a limit. Therefore, Eq.~\eqref{eq:local-optimality} must hold (otherwise, we can show by contradiction that $J_{t:T} (\boldsymbol{\pi}^{(k)})$ has not converged).

\end{document}